\documentclass[12pt]{article}
\usepackage{amsmath}
\usepackage{graphicx}
\usepackage{natbib}
\usepackage{url} 
\usepackage{booktabs}       
\usepackage{amsfonts}       
\usepackage{nicefrac}       
\usepackage{microtype}      
\usepackage{algorithm}
\usepackage{multirow}
\usepackage{subcaption}

\usepackage{algpseudocode}
\usepackage{amsthm}
\usepackage{xcolor}
\usepackage{bbm}
\usepackage{rotating}
\newcommand{\blind}{0}

\theoremstyle{plain}
\newtheorem{theorem}{Theorem}
\newtheorem{lemma}[theorem]{Lemma}

\theoremstyle{definition}
\newtheorem{definition}{Definition}

\makeatletter
\newcommand{\customlabel}[2]{%
\protected@write \@auxout {}{\string \newlabel {#1}{{#2}{}}}}
\makeatother

\DeclareMathOperator{\abs}{\text{abs}}

\begin{document}

\def\spacingset#1{\renewcommand{\baselinestretch}%
{#1}\small\normalsize} \spacingset{1}


\if0\blind
{
  \title{\bf Estimating Product Cannibalisation in Wholesale using Multivariate Hawkes Processes with Inhibition}
  \author{Isabella Deutsch\thanks{
    \texttt{isabella.deutsch@ed.ac.uk}}\hspace{.2cm} \\
    School of Mathematics, University of Edinburgh\\
    and \\
    Gordon J. Ross \\
     School of Mathematics, University of Edinburgh}
  \maketitle
} \fi

\if1\blind
{
  \bigskip
  \bigskip
  \bigskip
  \begin{center}
    {\LARGE\bf Title}
\end{center}
  \medskip
} \fi

\bigskip
\begin{abstract}
Product cannibalisation in the marketplace refers to the decrease in the sales of one product due to competition from another product. We examine this phenomenon in a wholesale data set provided by an international company. We use a multivariate Hawkes process where each product is represented by a dimension, with cross-inhibition effects that model product cannibalisation. To implement the Hawkes process with inhibition we resolve challenges regarding the integration of the intensity function and introduce a new, stronger conditions for stability as existing conditions are unnecessarily strict under inhibition. We conduct our analysis in a Bayesian framework, for which we design a dimension-independent prior on the cross-inhibition based on a reparametrisation. 
\end{abstract}

\noindent%
\vfill

\newpage
\spacingset{1.5} 

\section{Introduction} \label{sec_intro}

The purchase of one product can influence the sales volume of others in the same product category. For example, a consumer may choose to buy one particular good instead of a different (yet similarly designed) one. Such an effect is known as ``product cannibalisation'. More formally, product cannibalisation in the marketplace is defined as the decrease in sales of one product due to the sales of  a closely related product, or to the introduction of a new, similar product \citep{copulsky_cannibalism_1976}. In this paper we suggest using a multivariate Hawkes process with inhibition to estimate product cannibalisation in a Bayesian framework. We resolve issues around its estimation, including the integration of the intensity function, dimension-independent prior specification, and stability for Hawkes processes with inhibition.

Understanding product cannibalisation is important for a variety of market participants \citep{de_giovanni_product_2017}. For example, the producers of goods can use knowledge about product cannibalisation to improve product catalogues \citep{child_smr_1991, desai_quality_2001}, wholesalers can optimise for price \citep{de_giovanni_product_2017}, and retailers can make more informed decisions about which products to display on the shop floor \citep{kong_cannibalization_2015}. However, there exists ``\textit{little factual knowledge about the potential of market cannibalisation`}` \citep{atasu_so_2010}. 

Multiple approaches have been suggested to describe the dependence structure between product sales. A popular model comes from \citet{ruiz_shopper_2017} who consider the sequential choices of individual shoppers and how items can interact with each other. Other authors focus on similar articles that differ mainly in their quality \citep{desai_quality_2001}. In this context \citep{ghose_internet_2006} examine the impact of used books on online book sales. At each book purchase, a consumer can opt for a brand-new item or a pre-owned one. For example, it is of interest at what price point a consumer switches to buying a used version of a title, hence cannibalising the sales of new books. \citet{okorie_triple_2021} provide a meta-analysis for product cannibalisation focusing on re-manufactured goods in the circular economy. 

Several of the examples mentioned above use simple models driven by marketing theory \citep{atasu_so_2010} or game theory \citep{de_giovanni_product_2017}. A different approach comes from \citet{kamakura_predicting_1984}, who include product cannibalisation in their probabilistic choice models that tries to estimate the utility of each article. More recently, \citet{guidolin_has_2020} employ the Lotka-Volterra equations, a pair of non-linear differential equations commonly used in predator-prey scenarios, to capture the sales of the Apple iPhone, once the iPad was introduced. This models potential cannibalisation, e.g. a person buys an iPad rather than an iPhone or vice versa. Notably, this approach allows for asymmetric competition, e.g. more iPhone sold do not lead to less sales of iPads, but more iPad sales do lead to fewer iPhone sales. \citet{kong_cannibalization_2015} use a logit regression to model the sales of an article depending on a variety of covariates, such as availability and display inside the store. Crucially, they also include product cannibalisation (based on price and similarity) in their model.

The literature focuses on product cannibalisation for cumulative sales numbers instead of taking the temporal nature of (repeated) purchases into account. In recent years, a few notable exceptions have appeared. For example, \citet{aguilar-palacios_causal_2021} use a causal time series model to quantify product cannibalisation in grocery sales. Machine learning techniques have also very recently been employed in this field of study. For example, \citet{bekal_xgboost-based_2021} use boosting to predict sales in the presence of product cannibalisation, while \citet{garnier_concurrent_2022} rely on neural networks to uncover the ``competition'' between time series. 

Most of the approaches from the literature assume \textit{retail} data, i.e. good sold to the general public \citep{us_census_bureau_section_2011}. For example, \citet{garnier_concurrent_2022} record when a individual consumer purchased a particular freezer. In a business context data may only be available at a \textit{wholesale} level. Wholesale trade is defined as selling goods, often in large quantities, to other businesses who then might sell to the consumer \citep{us_census_bureau_section_2011}. In the following we refer to such a business who purchases items for retail distribution as a \textit{wholesale customer}.

In contrast to the above approaches, we propose to model cannibalisation using a multivariate point process, where each product is represented by one dimension, with the sales being the events. Point processes have been successfully used to model purchasing behaviour \citep[see, for example, ][]{pitkin_bayesian_2018}, but to the best of our knowledge, it is the first time that they are employed to estimate product cannibalisation.  More specifically, we model sales using a \textit{Hawkes process}. Classic Hawkes processes \citep{hawkes_spectra_1971} are point processes that describe the self-exciting behaviours of events. They are used to model events that occur in clusters or bursts, where one event makes it more likely that another event is happening soon after. Hawkes processes have been successfully applied in many application domains such as seismology  \citep{ogata_statistical_1988}, crime and terror modelling \citep{mohler2013modeling, Shelton, tucker_handling_2019}, accidents \citep{kalair_non-parametric_2021}, stock market trading \citep{rambaldi_role_2017}, and social media analytics \citep{lai_topic_2016}. 

While most applications of the Hawkes process only modelling excitation (i.e. the occurrence of one event makes future events more likely), Hawkes processes can also be used to capture inhibition (i.e. the occurrence of one event makes future events less likely) which is more relevant to the product cannibalisation setting. A prominent application of Hawkes processes with inhibition can be found in neural spike trains, where the inhibition captures the period of decreased activity after a neural spike event \citep{eichler_graphical_2017}. 

Inhibition is a substantial extension to the general Hawkes process, which bring additional challenges and subtleties to the estimation procedure. In this paper we address issues arising from a potentially negative intensity function (Section~\ref{subsec_nonnegative}), that then feed into problems when evaluating the likelihood as the intensity function needs to be integrated (Section~\ref{subsec_integrate_intensity}). We also examine conditions used to determine stability, which are unnecessarily strict when inhibition is present (Section~\ref{sec_stability}). Therefore we introduce a new condition that is stronger than both currently used conditions. These deliberations are not limited to the product cannibalisation example, but are relevant for all applications of Hawkes processes with inhibition.

The remainder of this paper is organised as follows. Section~\ref{sec_data} introduces the data. Section~\ref{sec_background_Hawkes} reviews the $1$ and $M$-dimensional Hawkes process and describes inhibition and the related challenges due to a potentially negative intensity function. In Section~\ref{sec_hawkes_for_product} we introduce our complete model used to estimate product cannibalisation. The estimation, using Stan, is then discussed in Section~\ref{sec_estimation}. In Section~\ref{sec_priorchoice} we examine the priors used for all parameters. Particular consideration is given to the influence parameter, for which we design a dimension-independent prior given a reparametrisation based on the branching structure. Section~\ref{sec_stability} contains theoretical consideration regarding stability in Hawkes processes with inhibitions. Section~\ref{sec_application_real} gives real world examples of product cannibalisation, first a two-dimensional example, then a four-dimensional example.



\section{Data} \label{sec_data}

We are collaborating with a major international company, anonymised as \textit{CompanyCo}, who have given us access to some of their data. We are not able to publish the company's name or its industry due to non-disclosure agreements. Crucially, this data is recorded on a wholesale level. Hence, it is not possible to infer when an individual consumer bought an item, but we know when a wholesale customer has placed an order with CompanyCo containing a certain article. 

Each order placed by a wholesale customer at CompanyCo is recorded as one line in the data and consists of a certain quantity of exactly one article. If multiple articles are ordered at the same time, then there will be a distinct entry per article. Each entry contains information on the article, the wholesale customer, and the order itself (e.g. the date on which the order was placed, and the quantity ordered). For this analysis we focus on the arrival of orders without additional covariates. These order dates are recorded in discrete time where only the day of the order is known. We add random noise with distribution $\mathcal{U}(0,1)$ to each event to be able to use a continuous time point process. However, if events (i.e. purchases of different articles) happened on the same day in the original data set, they are also set to the same (continuous) time in the altered data set. This avoids the introduction of spurious excitation as events that happen at the same day cannot influence each other.

The products sold by CompanyCo are goods used by individual consumers. Items that differ in appearance, but otherwise have the same design, are classified as different articles.  There are also labels within the CompanyCo brand that differ in price point or cater to particular consumer groups. One example of such a label is anonymised as \textit{SomeLabel}. For this analysis we examine two categories of products from CompanyCo's portfolio, which we call product class A and product class B.

For our application we will focus on the orders placed by one wholesale customer anonymised as \textit{BusinessGroup}. This particular wholesale customer was chosen by CompanyCo for in-depth analysis due to their medium sized and limited purchasing power such that they need to make active selections on which articles to stock. 

First, we examine when BusinessGroup is placing orders that contain products from one product class with CompanyCo. This is the same product class which examine in Section~\ref{sec_app_M4_B}. In Figure~\ref{fig_product_8} we plot the orders placed for eight articles over the course of 18 months. We observe a distinct seasonality in the orders of BusinessGroup. According to CompanyCo this is a typical pattern driven by the preorders and reorders of goods in a year, which is prevalent in most wholesale customer groups. 

\begin{figure}[htp]
    \centering
    \includegraphics[width=.9\linewidth]{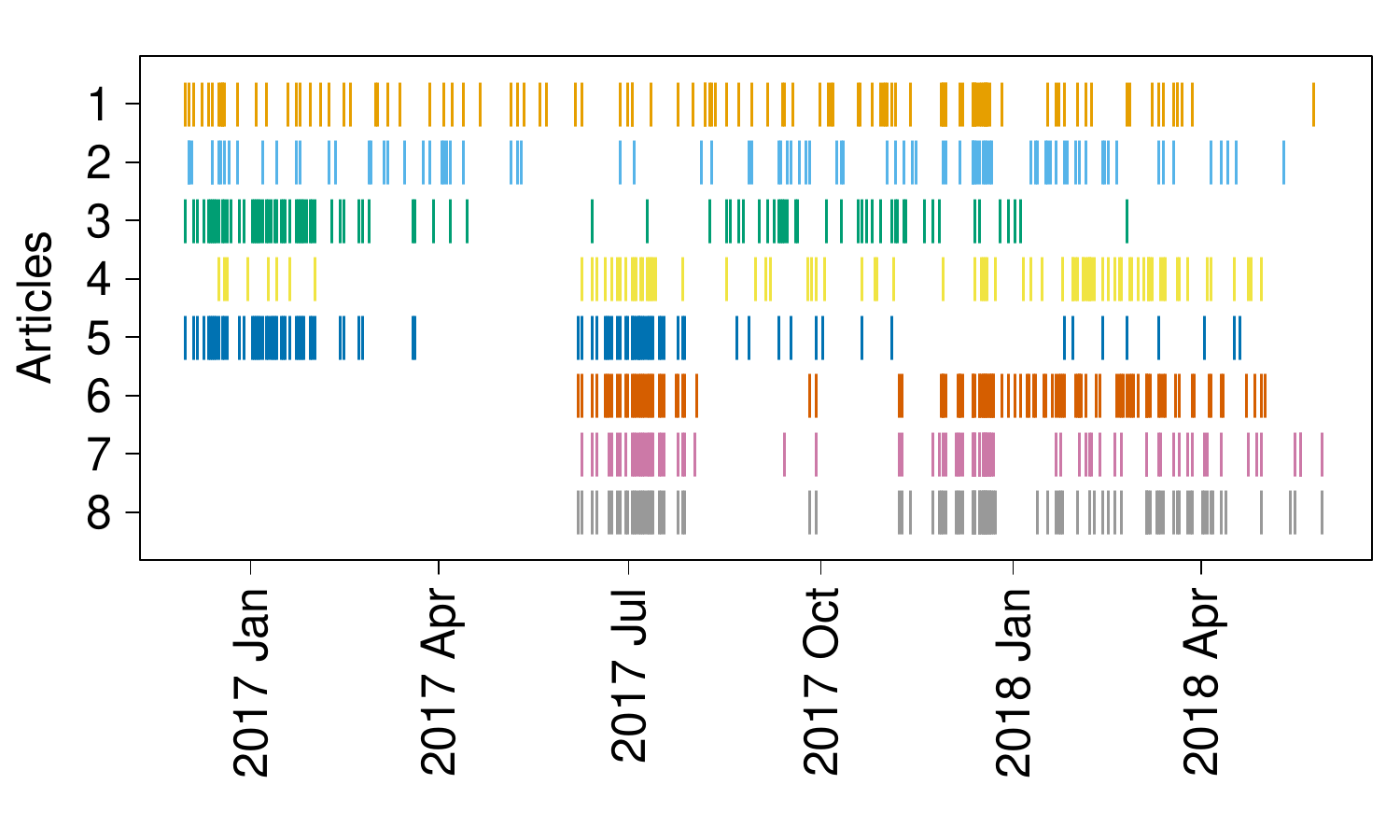}
    \caption{Orders placed for eight products from product class B by BusinessGroup. Each vertical bar indicates that on the particular day an order was placed that included the respective article.}
    \label{fig_product_8}
\end{figure}

There are three patterns that can be observed: monthly and weekday seasonality, and a Christmas indicator (defined as $24^{th}$ to $27^{th}$ of December). Figure~\ref{fig_monthly} displays the number of orders for products in product class B placed by BusinessGroup per month over $1.5$, years which gives a clear indication for a monthly variation. Figure~\ref{fig_weeklychristmas} plots the daily number of orders for December 2016 and January 2017. The plot shows that the number of orders is varies greatly between weekdays. In addition, there is a sharp drop in orders over the Christmas period, irrespective of the weekday and the otherwise large order numbers in December. Section~\ref{subsec_est_backgr} describes how these three characteristics are utilised in the background rate of our model. This seasonal pattern is also prevalent across product classes.

\begin{figure}[tp]
    \centering
    \includegraphics[width=.99\linewidth]{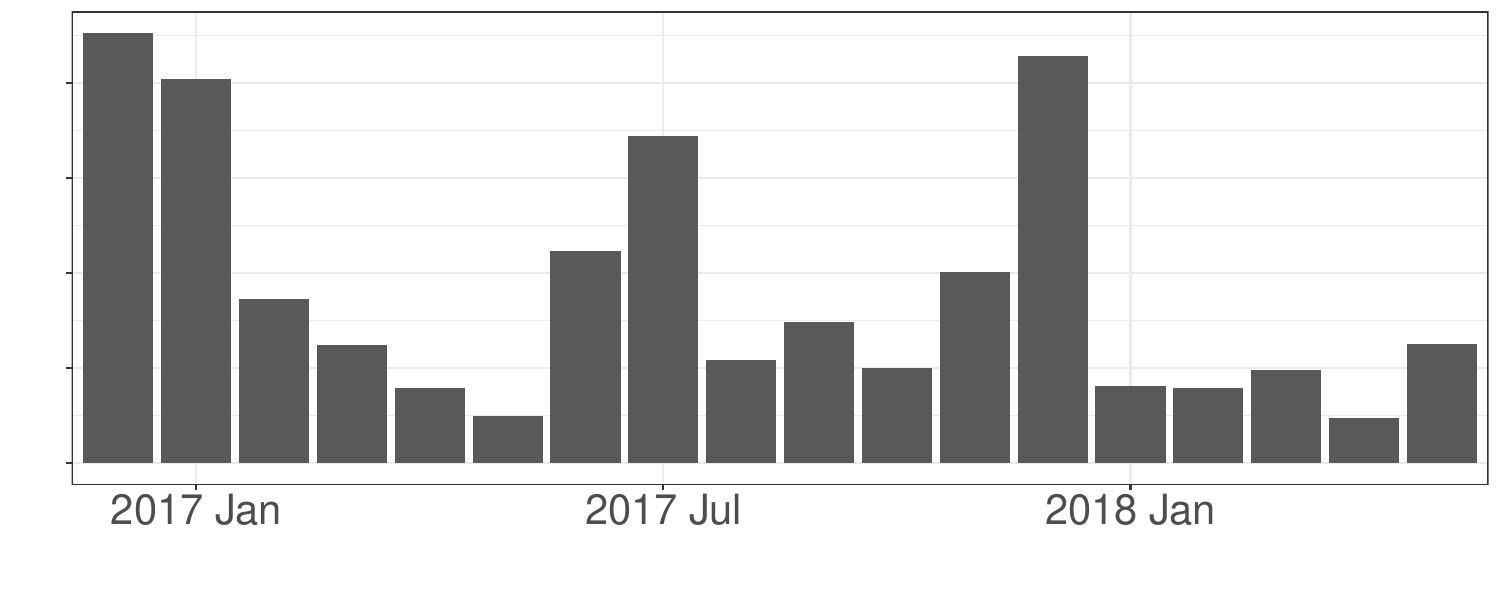}
    \caption{Orders placed each month for products in product class B by BusinessGroup. The y-axis has been removed for data protection reasons.}
    \label{fig_monthly}
\end{figure}

\begin{figure}[tp]
    \centering
    \includegraphics[width=.99\linewidth]{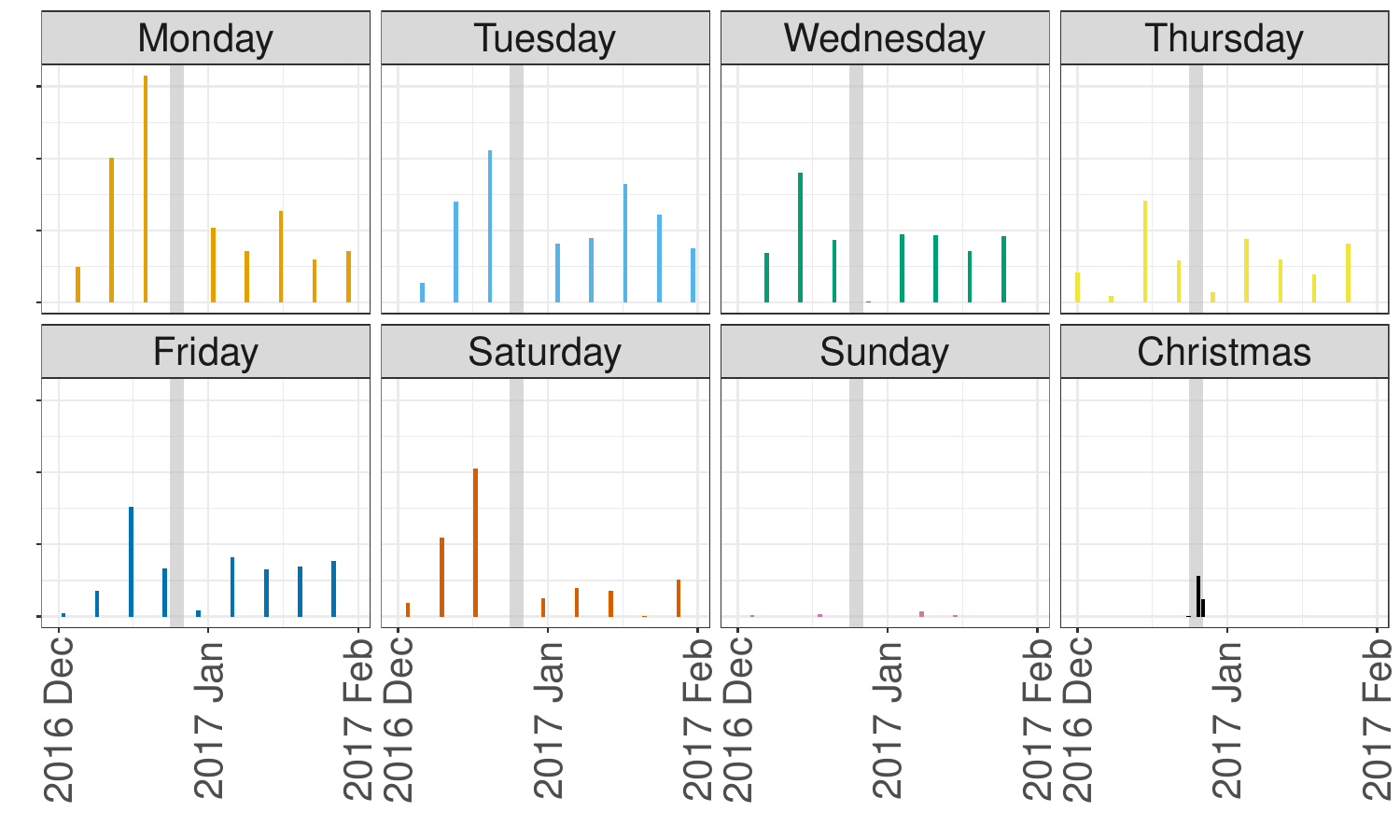}
    \caption{Total number of orders placed for products in product class B by BusinessGroup for December 2016 and January 2017 per day of the week and Christmas period. The shaded area highlights the Christmas period in each subplot. The y-axis has been removed for data protection reasons.}
    \label{fig_weeklychristmas}
\end{figure}

We now examine individual articles to motivate our search for product cannibalisation. Figure~\ref{fig_products2} plots the arrival of orders by BusinessGroup for two similar products. Apart from the seasonal variability two trends are visible. Firstly, orders for each articles are placed in rapid succession of themselves, i.e. when Article 1 is ordered it becomes more likely that the same article will be ordered again soon. The same holds true for Article 2. This \textit{self-excitation} is layered on top of the seasonal trends discussed above. Secondly, the point process displays \textit{cross-inhibition} where the order of Article 1 makes it less likely that Article 2 is purchased. Interestingly, this relationship also holds in reverse where a purchase of Article 2 makes it less likely that Article 1 is ordered. Given the similarities of the two articles (appearance, suggested retail price, no particular label) it is sensible to interpret this inhibition as product cannibalisation. 

\begin{figure}[tp]
    \centering
    \includegraphics[width=.9\linewidth]{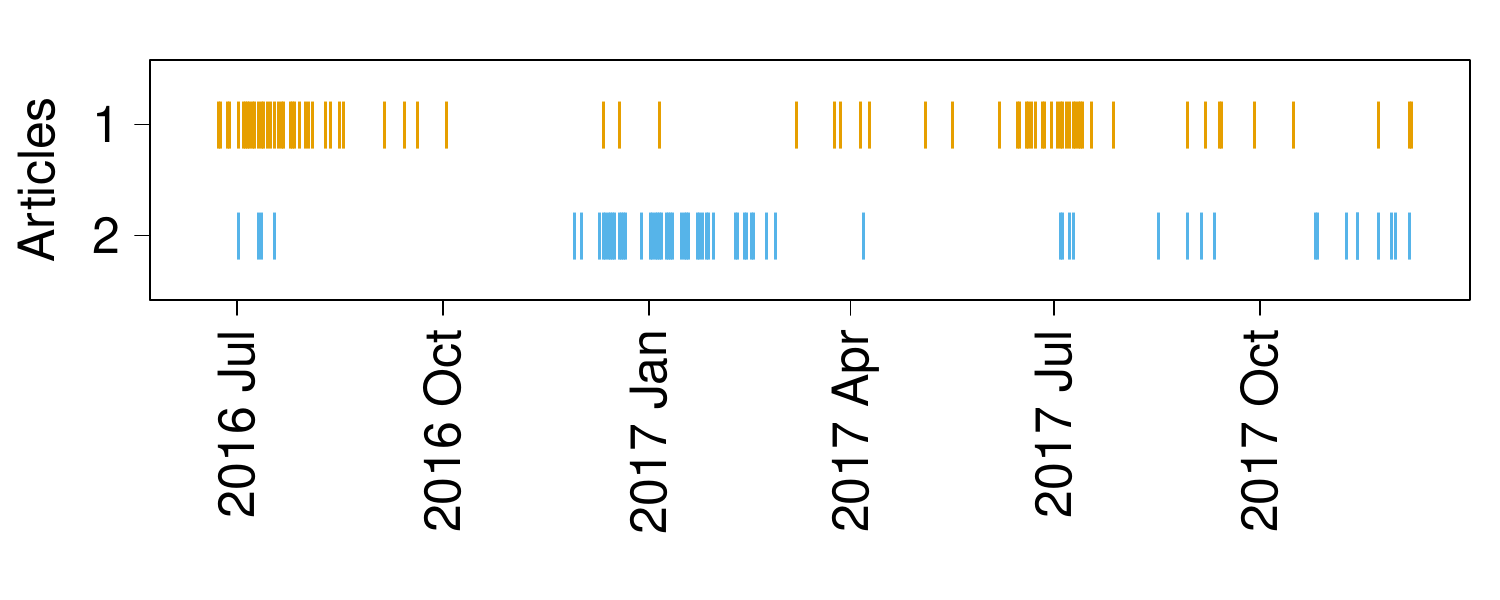}
    \caption{Orders placed for two products by BusinessGroup. Each vertical bar indicates that on the particular day an order was placed that included the respective article.}
    \label{fig_products2}
\end{figure}


\section{Hawkes Processes} \label{sec_background_Hawkes}

To estimate product cannibalisation we choose a multivariate Hawkes process where each article is represented by a dimension and we interpret cross-dimensional inhibition as product cannibalisation. Therefore we now examine the Hawkes process in more detail. This section considers the $1$ and $M$-dimensional Hawkes process, gives more detail on inhibition, and summarises the branching interpretation of the Hawkes process. The latter will then be utilised in a reparametrisation on which we base our prior (Section~\ref{sec_priors_Kstar}).

\subsection{Univariate Hawkes Process}

We first review the univariate (= $1$-dimensional) Hawkes process. For each event from a point process we record the time when it happened $Y = \left(t_1 \dots t_N \right) \in \left[0, T_{max}\right]$, such that $t_i \in \mathbb{R}^+$ is the event time at which the $i^{th}$ event took place. The self-exciting linear Hawkes process is defined by its conditional intensity function $\lambda(t | Y^t, \Theta)$, which at time $t$ is conditional on the previous events $Y^t = \{t_i : t_i < t \}$ \citep{hawkes_spectra_1971} and parameters $\Theta$ used to specify the parametric form of intensity. For convenience of notation the dependence on $Y^t$ and $\Theta$ is suppressed further on. 

The following specification encompasses the self-exciting nature: 
\begin{equation} 
    \lambda(t| Y^t, \Theta) = \mu(t|\Theta) + \sum_{i: t_i < t} K \, g(t - t_i| \Theta)
    \label{eqn:hawkes_simple}
\end{equation} 
Here, $\mu(\cdot) > 0$ is the background rate that can capture seasonality and underlying trends. We call $g(\cdot)$ the influence kernel, where $g(x) \geq 0$ for $x \geq 0 $ and $\int_0^\infty g(x) \, dx = 1$. This decides how much the influence is spread out over time, whereas $K$ captures the overall magnitude of the influence. A classic Hawkes process restricts $K \geq 0$, which only allows excitation. When $K \leq 0$, which corresponds to inhibition, further considerations are required. These are discussed in Section~\ref{subsec_inhibition}.

Each observation prior to $t$ contributes to the intensity at time $t$ as governed by the kernel $g(\cdot)$ and $K$. This drives the self-exciting behaviour of the Hawkes process. Both the background rate and the kernel depend on parameters $\Theta = (\theta_\mu, K, \theta_g)$ where $\theta_\mu$ and $\theta_g$ contains all parameters from $\mu(\cdot)$ and $g(\cdot)$, respectively. To estimate them, a plethora of methods is successfully employed in the literature, for example maximum likelihood approaches and Bayesian methods \citep{Veen_2008_estimation, rasmussen_bayesian_2013, chen_2018_direct, ross_bayesian_2021}. 

\subsection{Multivariate Hawkes Process} \label{subsec_Mdim}

The model presented in Equation~\ref{eqn:hawkes_simple} extends to an $M$-dimensional linear Hawkes process incorporating both self-excitation and cross-excitation. This is called a multivariate linear Hawkes process, which is an essential building block to estimate product cannibalisation. Assume that there are $M$ dimensions with event times $Y_1 = \left(t_{1\,1} \dots t_{1\,N_1} \right)$ in dimension $1$ to event times $Y_M = \left(t_{M\,1} \dots t_{M\,N_M} \right)$ in dimension $M$. At time $t$ the intensity in dimension $i$ is:
\begin{equation}
    \lambda_i(t) = \mu_i(t) + \sum_{j = 1}^M\sum_{l:  t_{j\,l} < t} K_{ji}\, g_{ji}(t-t_{j\,l}) 
    \label{eqn:hawkes_general}
\end{equation}

We assume the following form for the excitation kernel for all $i,j$: $g_{ij}(x) \geq 0$ for $x \geq 0 $ and $\int_0^\infty g_{ij}(x) \, dx = 1$. Here, we still assume $ K_{ij} \geq 0$, describes the excitation effect an event in dimension $i$ has on dimension $j$. The case of inhibition ($ K_{ij} \leq 0$) is discussed in Section~\ref{subsec_inhibition}.  We write $\mathbf{K} = \{K_{ij}\}$ where $i,j = 1 \dots M$. Note that we do not enforce symmetry or any other structure in $\mathbf{K}$. Here, $\theta_\mu$ contains all parameters from the background rates $\mu_i(\cdot)$, and $\theta_g$ from $g_{ij}(\cdot)$ for all $i,j = 1 \dots M$.

\subsection{Inhibition} \label{subsec_inhibition}

When a $K_{ij} \leq 0$ in Equation~\ref{eqn:hawkes_general}, this is called inhibition. This implies that an event in dimension $i$ decreases the intensity function of dimension $j$, hence making it less likely that an event in dimension $j$ takes place. In our proposed model in Section~\ref{subsec_model} we permit both excitation and inhibition, implying $K_{ij} \leq 1$ for $i,j = 1 \dots M$. This subtle change leads to a variety of additional considerations that need to be made compared to the excitation-only Hawkes process. 

One concern of Hawkes processes with inhibition is the necessity of a non-negative intensity function. Any point process requires, by design, that the  intensity function at every $t \in [0, T_{max}]$ is non-negative. When a $K_{ij}$ in Equation~\ref{eqn:hawkes_general} is negative it is not guaranteed that the intensity always stays non-negative. We discuss two approaches on how to handle this issue in Section~\ref{subsec_nonnegative}, as well as resulting implication for the likelihood in Section~\ref{subsec_integrate_intensity}. In addition, we also found that the commonly used conditions to assess stability were too restrictive when inhibition was present. Section~\ref{sec_stability} describes this issues in detail and proposes an adapted condition that is stronger than the previously used ones. These deliberations are not limited to the product cannibalisation example, but are relevant for all applications of Hawkes processes with inhibition.

\subsubsection{Non-Negative Intensity} \label{subsec_nonnegative}

There are two main approaches in the literature to ensure non-negativity: restricting the parameter space or using a link function to ensure positivity. For the remainder of the paper we choose the link function approach.

\paragraph{Restricting the Parameter Space}

One common, although crude, way to guarantee a non-negative intensity function is to restrict the parameter space based on the observed data set. In a Bayesian framework this can be incorporated into the prior by attributing a non-zero probability to parameter combinations that give a non-negative intensity evaluation everywhere for the data at hand. However, this method has three fundamental shortcomings: potential non-consistency, sampling issues, and data-dependency.

First, it is typical that the intensity function of a Hawkes process exhibits many spikes and drops with short, rapid changes. Excluding all parameter combinations that lead to a decrease below zero \textit{somewhere} would substantially restrict the parameter space, where many parameter combinations are prohibited. This could potentially exclude the true parameters if inhibition is truly present, which would lead to a non-consistent estimation procedure. Second, restricting the parameter space can cause sampling issues as MCMC samplers are prone to boundary artefacts under such confining conditions. Finally, the `permissibility' of a parameter combination depends on the data. A set of parameters causing a non-negative intensity on a data set does not guarantee this property when more data is collected. We therefore do not recommend this approach and instead rely on a link function, as described in the next section.

\paragraph{Link Function} \label{subsec_link}
Since restricting the parameter space to ensure a non-negative intensity function has clear draw backs we instead use a link function $\phi(\cdot)$. This is another common practice and leads to the following intensity:
\begin{equation}
    \lambda_i(t) = \phi \left(\mu_i(t) + \sum_{j = 1}^M\sum_{l: t > t_{j\,l}} K_{ji}\, g_{ji}(t-t_{j\,l}) \right)
    \label{eqn:hawkes_general_link}
\end{equation}

The link function $\phi(\cdot): \mathbb{R} \to \mathbb{R}^+$ ensures a non-negative intensity at every $t$. For example, \citet{mei_neural_2017} use the \textit{softplus} function $\phi(x) = s \, \, log(1 + exp(x/s))$ with parameter $s$. Another straightforward choice of $\phi$ is the \textit{ReLU} function where $\phi(x) = \max(a, x)$ for a small, non-negative $a$. A popular approach in the literature is to set $a=0$ \citep{lemonnier_nonparametric_2014, lu_high-dimensional_2018, costa_renewal_2020}. Crucially, this choice of a link function preservers the interpretation of $K_{ij}$ as the average number of direct offsprings (as described in Section~\ref{subsec_branching}). Hence, for the remainder of this paper we use the \textit{ReLU} function with $\phi(x) = \max(0, x)$.


\subsection{Branching Interpretation} \label{subsec_branching}

We now examine how data from a one-dimensional linear Hawkes process can be related to an underlying branching structure as we utilise this concept in Section~\ref{sec_offsprings} to motivate a new parametrisation and subsequent prior specification for product cannibalisation. The extension to multivariate Hawkes processes is straightforward and, for example, discussed in \citet{embrechts_multivariate_2011}. Here we consider the common branching structure interpretation when $K > 0$. 

The model in Equation~\ref{eqn:hawkes_simple} can be written as the superposition of Poisson processes \citep{hawkes_cluster_1974} such that the intensity function is a sum of independent Poisson processes. Suppose that $j$ events $(t_1 \dots t_j)$ happened before time $t$. Then the intensity at time $t$ is the sum of the background process $\mu(t)$ and $j$ offspring processes with intensities $K \, g(t-t_j)$, where each offspring process was \textit{triggered} by a previous event. This gives rise to the following branching structure interpretation of a Hawkes process \citep[][p. 202]{daley_introduction_2003}. Data generated from an excitation-only Hawkes process consists of two types of events that come from distinct Poisson processes, where each event is generated by exactly one process \citep{rasmussen_bayesian_2013}:

\begin{enumerate}
    \item \textit{Immigrant events} come from the background process with intensity $\mu(t)$. 
    \item \textit{Offspring events} come from an offspring process which had been triggered by a previous event. Here, each event has an average of $K$ \textit{direct} offsprings if $T_{max} \to \infty$. 
\end{enumerate}

Note that offspring events trigger offspring processes as well. This can lead to \textit{cascades} started by an immigrant event that has offspring events, which, in turn, has offspring events etc. Figure~\ref{fig_branching} visualises this branching structure. This concept will be utilised in Section~\ref{sec_offsprings} to introduce a new parametrisation based on the total number of offsprings.

\begin{figure}[tp]
    \centering
    \includegraphics[width=.8\linewidth]{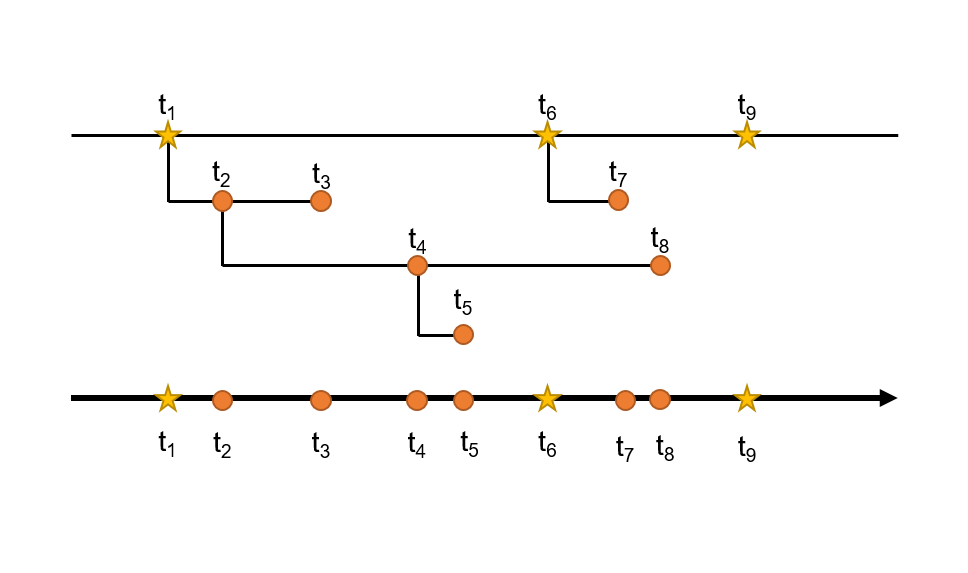}
    \caption{Sketch of the branching structure interpretation in one dimension. In this example $(t_1, t_6, t_9)$ are the immigrant events, indicated by a star on the top and the summarised timeline below. The event $t_1$ has two \textit{direct} offsprings, namely $t_2$ and $t_3$. The former has further offsprings. For example $t_1$ has a total of five offsprings (number of offspring events in the cascade that started in $t_1$), whereas $t_9$ has none.
}
    \label{fig_branching}
\end{figure}

While it is not known whether an events is an immigrant of offspring for a given data set, this perspective nevertheless can be exploited for sampling, inference, and interpretation \citep{rasmussen_bayesian_2013,  ross_bayesian_2021}. These interpretations hold asymptotically for $T_{max} \to \infty$. For finite $T_{max}$ some offspring events may be larger than $T_{max}$ and therefore would not be included in a simulated data set. Such edge effects are common in the Hawkes literature and diminish when $T_{max}$ is large \citep[][p. 275]{daley_introduction_2003}. 

This interpretation extends to the excitation multivariate case (all entries of $\mathbf{K}$ non-negative) where an event in dimension $i$ triggers offspring processes in all dimensions $j$. In turn, resulting offspring events can trigger offspring processes in all dimensions. Hence, cascades can potentially lead through different dimensions, which all contribute to the total number of offsprings of the immigrant event that started the cascade, which is discussed in Section~\ref{sec_offsprings}.


\section{Hawkes Processes for Product Cannibalisation} \label{sec_hawkes_for_product}

We are now at the point to define the model we use to estimate product cannibalisation in wholesale. This section contains a description of our modelling approach and its intensity function. In addition, we give details on the choice of the background rate $\mu_i(\cdot)$ and the influence kernel $g_{ij}(\cdot)$.

\subsection{Model} \label{subsec_model}
To estimate product cannibalisation we use the multivariate Hawkes process such that each article $i = 1 \dots M$ is represented by a dimension  $i = 1 \dots M$. Whenever an order for article $i$ is placed we record an event in dimension $i$. Based on this data we use the following intensity function for our model:
\begin{equation}
    \lambda_i(t) = \left[\mu_i(t) + \sum_{j = 1}^M\sum_{l: t > t_{j\,l}} \mathbf{K}_{ji}\, g_{ji}(t-t_{j\,l}) \right]_+
    \label{eqn:hawkes_intensity_model}
\end{equation}
where $\left[ \cdot \right]_+$ indicates that any negative evaluation inside the brackets is set to zero, as discussed above in  our choice of link function.

The only restriction we place on the influence is that $K_{ji} < 1$, which means that both excitation and inhibition are permitted in this model. If the interaction $K_{ji}$ for $i \neq j$ is negative we interpret this inhibition as product cannibalisation as the occurrence of an event in dimension $i$ (article $i$ is bought) makes it less likely that an event happens in dimension $j$ (article $j$ is ordered). 

Two parts of the intensity function in Equation~\ref{eqn:hawkes_intensity_model} still need to be defined, $\mu_i(\cdot)$ and $g_{ij}(\cdot)$. The subsequent sections provide  details on both the background rate and the influence kernel to conclude the definition of our model.

\subsection{Background Rate} \label{subsec_model_backgr}
The choice of background rate $\mu_i(\cdot)$ for Hawkes processes is often flexible and application-specific. For example, \citet{mohler2013modeling} use a Log-Gaussian Cox process, \citet{molkenthin_gp-etas_2022} utilise a Gaussian Process to represent the background rate, and \citet{kolev_semiparametric_2020} employ a Dirichlet process. As the data described in Section~\ref{sec_data} displays distinct seasonal variability this needs to be taken into account in our estimation procedure through an adapted background rate. We choose for $i = 1 \dots M$
\begin{equation}
    \mu_i(t) = c_i \, b(t)
\end{equation}
where $b(\cdot)$ is accounting for general seasonality. The article-specific scaling parameter $c_i$ is flexible enough while while remaining computationally cheap. We use the following parametric form for seasonal part of the background rate.
\begin{equation}
    b(t) = 
    \begin{cases}
    \left[ 
    \varphi_1 \mathbbm{1}_{\text{Mon}}(t) + \dots + \varphi_7 \mathbbm{1}_{\text{Sun}}(t)
    \right]  \left[ 
    \varphi_8 \mathbbm{1}_{\text{Jan}}(t) + \dots +  \varphi_{19} \mathbbm{1}_{\text{Dec}}(t)
    \right] ,  & \text{if } \mathbbm{1}_{\text{Christmas}}(t) = 0\\
    \varphi_{20},              & \text{if } \mathbbm{1}_{\text{Christmas}}(t) = 1
\end{cases}
\end{equation}

This describes a multiplicative effect between day of the week ($\varphi_1 \dots \varphi_7$) with the month ($\varphi_8 \dots \varphi_{19}$) outside of the Christmas period ($24^{th}$ till $27^{th}$ of December), and a constant rate $\varphi_{20}$ during the Christmas period.

\subsection{Influence Kernel} \label{subsec_model_Influence_Kernel}
The options for influence kernels are equally broad. In some instances they are highly problem-specific, for example \citet{browning_simple_2021} use a histogram kernel for Covid modelling. A popular choice of the influence kernel is the exponential kernel, which performs well in many examples \citep[e.g. ][]{blundell_modelling_2012, Shelton}. For the influence kernels we utilise this exponential kernel
\begin{equation}
    g_{ij}(x) = \beta_{ij} \, exp \left( -\beta_{ij} \, x \right)
    \label{eqn:hawkes_influence kernel}
\end{equation}
with $\beta_{ij} > 0$ for $i,j = 1 \dots M$. Additionally, we assume that all $\beta_{ii} = \beta_{\text{diag}}$ and $\beta_{ij} = \beta_{\text{off}}$ when $i \neq j$. Hence, the values for $\beta_{ij}$ are the same for all self-influences, as well as cross-influences, respectively.






\section{Estimation} \label{sec_estimation}

To estimate this model from Section~\ref{sec_hawkes_for_product} we use a Bayesian approach incorporating the priors that will be discussed in Section~\ref{sec_priorchoice}. We utilise standard Stan \citep{Stan_RStan_2019} to obtain posterior samples and analyse the posterior distribution through density plots and summary statistics generated from these samples. The following sections state the likelihood and discuss challenges arising from the need to integrate the intensity function when evaluation the likelihood. We describe different proposed solutions for this from the literature, as well as the approximation used going forward. Finally, we specify the estimation of the background rate. 

\subsection{Likelihood} \label{subsec_likelihood}

For any multivariate point process the likelihood for data $Y$ is
\begin{equation}
    p(Y | \Theta) = \prod_{m = 1}^M \prod_{i = 1}^{N_m} \lambda_m(t_i | \Theta ) \, exp\left(-  \Lambda \right)
 \label{eqn:likelihood}
\end{equation}
 \citep[][p. 23]{daley_introduction_2003} where $\Lambda = \sum_{m = 1}^M \int_0^{T_{max}} \lambda_m(x | \Theta ) \, dx$ and $\Theta = (\theta_\mu, \mathbf{K}, \theta_g)$.  Note that the evaluation of both the likelihood and the log-likelihood require $\Lambda$, which is the integral of the intensity function. While this is straightforward task in an excitation-only Hawkes process, more deliberation is required for the inhibition-encompassing model introduced in Section~\ref{subsec_model}. The following section reviews the occurring issues and presents an exact solution for a special case of the parameter structure (details in Appendix~\ref{appendix_roots}), and an approximate one applicable in the general case (with additional material in Appendix~\ref{appendix_integral}).

\subsubsection{Integrating the Intensity} \label{subsec_integrate_intensity}

Many parameter estimation methods, such as our subsequent Bayesian procedures, need to evaluate the likelihood, which requires $\Lambda$.  We now discuss different approaches from the literature, present exact solutions for particular choices of the influence kernel, and suggest an approximate solution that will be subsequently utilised. 

When all $K_{ij}$ are non-negative and $\phi(x) = x$, the integral of the intensity can be computed by integrating each segment between events (as well as the ones between $0$ and the first event and between the last event and $T_{max}$). This reduces to an easy computation of the integral that sums over the background rate and the contributions of each event \citep{ogata_lewis_1981}. However, it becomes difficult to integrate the intensity function under inhibition in this scenario as the intensity might drop below zero. \citet{lu_high-dimensional_2018} choose to calculate the integral as in the excitation-only case, even though the parts of the integral below zero contributing negatively to the integral.

The use of a link function $\phi(\cdot)$ guarantees a non-negative intensity function (see Section~\ref{subsec_link}). However, any choice of link that is not the identity prohibits a straightforward calculation of the integral as the approach by \citet{ogata_lewis_1981} cannot be employed anymore. \citet{mei_neural_2017} explain in their supplementary materials that they approximate $\Lambda$ for a one-dimensional Hawkes process by sampling a single $t^*$ uniformly from $[0, T_{max}]$ and then use $\hat{\Lambda} = T_{max} \lambda(t^*)$. While $\hat{\Lambda}$ is indeed an unbiased estimator for $\Lambda$ it has a large variance. The intensity functions often exhibit rapid ups and downs, such that just evaluating it once cannot capture all aspects of the function. \citet{ertekin_reactive_2015} use Approximate Bayesian Computation to circumvent the problem as this method does not require an evaluation of the likelihood and therefore the integral is not required.

Another possibility is to identify the intervals of the intensity function that are non-negative and integrate only those \citep{bonnet_maximum_2021}. While this is an exact solution, this approach requires finding the roots of the intensity function, which depend on the exact specifications of the model. In Appendix~\ref{appendix_roots_general} we show that the roots of the intensity function with the exponential kernel are the solutions to a high-order polynomial. The computational complexity of this root finding approach is high since the root-finding needs to be employed between each event time in every dimension. When setting all $\beta_{ij}$ to be equal, we find a simple expression for the roots when using the exponential kernel, as shown in Appendix~\ref{appendix_roots_beta}. However, this assumption may be too restrictive in application. 

Going forward we use a numerical approximation of $\Lambda$ without placing any restrictions on $\beta_{ij}$. In each segment between events in every dimension we use a cubic Simpson's rule approximation. While this is a costly approximation in terms of computational complexity, any lower level approximation introduced too much bias. For details see Appendix~\ref{appendix_integral}.

\subsection{Background Rate} \label{subsec_est_backgr}

We now consider the estimation of the background rate $\mu_i(\cdot)$ for $i = 1 \dots M$. As discussed in Section~\ref{subsec_model_backgr} the background rate consist of two parts
\begin{equation}
    \mu_i(t) = c_i \, b(t)
\end{equation}
where $c_i$ is a product specific scaling parameter and $b(\cdot)$ is captures general, seasonal trends. It contains a multiplicative effect between day of the week ($\varphi_1 \dots \varphi_7$) with the month ($\varphi_8 \dots \varphi_{19}$) outside of the Christmas period ($24^{th}$ till $27^{th}$ of December), and a constant rate $\varphi_{20}$ during the Christmas period. 

Estimation of components of the background rate can either be estimated jointly with all other model parameters \citep{mohler2013modeling, kolev_semiparametric_2020, molkenthin_gp-etas_2022} or outside of the model \citep{helmstetter_comparison_2006}. For our approach we choose to estimate the product-specific scaling parameter $c_i$ within Stan and provide the seasonal component $b(\cdot)$ as a plug-in estimate which is produced beforehand outside of Stan. For this we considering all articles in the respective category, for example all products in product class A for the first example in Section~\ref{sec_app_M2_A}. We take a sample of events without replacement from each article to ensure that articles with more sales do not dominate the estimation and then employ an optimiser in \textsf{R} on this data set to obtain the maximum likelihood estimation for $(\varphi_1 \dots \varphi_{20})$. These are then used to construct the plug-in estimate $b(\cdot)$, which is passed to Stan.

\section{Prior Choice} \label{sec_priorchoice}

To implement the proposed model from Section~\ref{subsec_model} for product cannibalisation in a Bayesian framework we investigate prior choices for $\mu_i(\cdot)$, $\mathbf{K}$, and $\beta$ (using an exponential kernel). 

\subsection{Prior for $c_i$} \label{subsec_priors_mu}

For each dimension $i= 1 \dots M$ we use $c_i>0$ with prior
\begin{align}
     c_i & \sim \mathcal{N}(0,3) & \text{ for } i = 1 \dots M 
\end{align}
that scales the plug-in estimates accordingly as the sale volumes may differ between articles.

\subsection{Prior for $\mathbf{K}^{*}$} \label{sec_priors_Kstar}

Throughout the literature Hawkes processes are parameterised in terms of $\mathbf{K}$. For example, \citet{browning_simple_2021} choose a uniform prior for each of its entries. However, when the dimension $M$ is large the non-negative entries of $\mathbf{K}$ have to be smaller in order to retain stability in accordance the criteria outlined in Section~\ref{sec_stability}. For example, when all entries of $\mathbf{K}$ are $0.4$, a two-dimensional process is stable, whereas a three-dimensional $\mathbf{K}$ has an eigenvalue larger than 1 and hence is not stable. Priors on $\mathbf{K}$ in a Bayesian framework would therefore have to be adapted according to the dimension $M$. Instead, we suggest to reparametrise the model to circumvent this problem using the total number of offsprings, which does not suffer from this dimension-dependency, and use this notion to place priors.

\subsubsection{Reparametrise the Model Using the Total Number of Offsprings} \label{sec_offsprings}
In this section we formally define the total number of offsprings, which is then used for the prior elicitation. For now, we assume that all entries of $\mathbf{K}$ are non-negative. Section~\ref{subsec_branching} introduces the notion of direct offspring, i.e. an event in dimension $j$ from the immigrant process of an event in dimension $i$. However, when estimating the parameters it can become difficult to distinguish between and event in $i$ triggering an immigrant process in dimension $j$ (`$i \to j$') and an event in $i$ triggering an immigrant process in dimension $k$, which produces an event that, in turn, triggers an immigrant process in $j$ (`$i \to k \to j$`). This problem is perpetuated further in higher dimensions \citep{eichler_causal_2013}. We therefore propose to investigate the total number of offsprings that circumvents this issue. For this, two kinds of events are define as an indirect offspring of event $t_i$:
\begin{enumerate}
    \item an event from a process that was triggered by a direct offspring of $t_i$
    \item an event from a process that was triggered by an indirect offspring of $t_i$
\end{enumerate}
We define $K^*_{ij}$ as the total number of offsprings an event in dimension $i$ has in dimension $j$, which is calculated as the sum of direct and indirect offsprings in dimension $j$. We write $\mathbf{K}^{*} = \{ K^*_{ij} \}$ were $i,j = 1 \dots M$.

\begin{theorem}
The total number of offsprings is 
\begin{equation}
    \mathbf{K}^{*} = \left(I - \mathbf{K} \right) ^{-1} - I   \nonumber
\end{equation}
\end{theorem}

The proof is available in Appendix~\ref{appendix_offsprings}. It is important to note that the idea and definition of $\mathbf{K}^{*}$ have already been presented in the literature, but not in a unified concept towards use in application. On one hand, \citet{bacry_2016_first} introduce this formula without its interpretation in the context of a matrix convolution. On the other hand, \citet{bacry_estimation_2016} define the concept of total offsprings but do not provide a closed form expression. Crucially, neither of them make extensive use of the concept, in particular not for prior elicitation. 

If a process is stable (see Section~\ref{sec_stability}), the calculation remains the same when entries of $\mathbf{K}$ are negative and $\mathbf{K}^{*}$ still provides meaningful interpretation. While positive entries of $\mathbf{K}^{*}$ describe the average number of total offsprings, a negative entry summarises the \textit{negative contributions} to the intensity function across dimensions. The number of actually inhibited events depends on the number of events in the process. Nevertheless, $\mathbf{K}^{*}$ retains its attractive interpretation and can be used to place priors without having to consider the dimensions $M$.

Crucially, the entries of $\mathbf{K}^{*}$ for a stable process are not dependent on the dimension $M$. We therefore reparameterise the multivariate Hawkes process in terms of $\mathbf{K}^{*}$ such that the intensity for dimension $i$ is written as: 

\begin{equation}
    \lambda_i(t) = \phi \left(\mu_i + \sum_{j = 1}^M\sum_{l: t > t_{j\,l}} \left\{ f(\mathbf{K}^*) \right\}_{ji}\, g_{ji}(t-t_{j\,l}) \right)
    \label{eqn:hawkes_lik_Kstar}
\end{equation}

where $f(\mathbf{X}) = I - (\mathbf{X} - I)^{-1}$.

\subsubsection{Normal Priors}  \label{subsec_normalpriors}
This reparameterisation in Equation~\ref{eqn:hawkes_lik_Kstar} permits us to use priors directly for $\mathbf{K}^{*}$ as these values do not depend on the dimension $M$. We restrict the parameter space of $\mathbf{K}^{*}$ such that only stable parameters (according to \ref{c3new} from Section~\ref{sec_stability}) are allowed. We place independent normal priors on each entry of $\mathbf{K}^{*}$:
\begin{align}
     K^*_{ij} & \sim \mathcal{N}(0,0.5) & \text{ for } i,j = 1 \dots M & \text{ (stable only) }
\end{align}
Note that we do not enforce symmetry or any other structure in $\mathbf{K}^{*}$.

\subsection{Prior for $\beta$} \label{subsec_priors_beta}

For the influence kernels we utilise the popular exponential kernel. As discussed in Section~\ref{subsec_model_Influence_Kernel}, we assume that all $\beta_{ii} = \beta_{\text{diag}}$ and $\beta_{ij} = \beta_{\text{off}}$ when $i \neq j$. Hence, the values for $\beta$ are the same for all self-influences, as well as cross-influences, respectively. For $\beta_{\text{diag}}$ and $\beta_{\text{off}}$ we choose the following priors.
\begin{align}
    \beta_{\text{diag}} &  \sim \mathcal{U}(0.05, 0.5) & & \\
    \beta_{\text{off}} & \sim \mathcal{U}(0.05, 0.5) & & 
\end{align}
As the data is measured in days, the lower bound of the prior ensures that influence of an event in dimension $i$ onto dimension $j$ is not too far in the future. For example for $ \beta_{\text{diag}} = 0.1$, the median of the exponential distribution is approximately $7$, which means that half of the influence of an event happens within a week of it. The upper bound of the distribution makes sure that the influence kernel is at least somewhat spread out and not concentrated immediately after an event.

\section{Stability} \label{sec_stability}

When using Hawkes processes one usually assumes, either explicitly or implicitly, \textit{stability}. This ensures that there are not infinitely many events happening. In a Bayesian statistics this can be easily accomplished by restricting the prior space. However, we found that the two stability conditions currently used in the literature are unnecessarily restrictive when inhibition is present. We therefore propose a new condition which is stronger than the two existing ones and use it in the prior specification (see Section~\ref{sec_priors_Kstar}). This section contains the derivation of our condition and compares it to the two existing ones. 

Due to the self-exciting behaviour of the Hawkes process it is possible that an infinite number of events take place in finite time. For example, this can happen in a one-dimensional Hawkes-process if each event has more on average one or more offsprings. This behaviour is called \textit{supercritical} \citep{helmstetter_subcritical_2002} or \textit{explosive} \citep{browning_simple_2021}. However, for real-life applications it can be desirable to limit the parameter space to non-explosive instances \citep{Kolev2019}. This is referred to as \textit{stability} \citep{bremaud_stability_1996, bacry_2015_sparse}.

\begin{definition}
A Hawkes process is stable if there exists a unique stationary distribution of the process with finite average intensity \citep{bremaud_stability_1996, sulem_bayesian_2021}.
\end{definition}

Two conditions (\ref{c1eigen}, \ref{c2colsum} defined below) have been used in the literature to determine stability, as outlines in Section~\ref{subsec_stability_literature}. In Section~\ref{subsec_newcond} we introduce a new condition that is stronger than both of the previously used ones. We prove that when at least one of (\ref{c1eigen}, \ref{c2colsum}) hold, so does our condition. Moreover, there exist parameters $\mathbf{K}$ for which our condition holds when neither of (\ref{c1eigen}, \ref{c2colsum}) do. This permits us to have a unified approach for checking stability and to classify a larger set of parameters as stable. We also provide a toy-example to illustrate this usefulness of our suggested condition.

We introduce the following notations: $\abs(\mathbf{A})$ is an $M \times M$ matrix where each entry is $|A_{ij}| $, the absolute value of $A_{ij}$. Moreover, $\mathbf{A}^+$ is the matrix with entries $\max(A_{ij},0)$. We write $\rho(\mathbf{A})$ for the spectral radius of matrix $\mathbf{A}$, i.e. the largest absolute eigenvalue of matrix $\mathbf{A}$.  

\subsection{Stability Conditions in the Literature} \label{subsec_stability_literature}

To start, we state two conditions to assess stability that have been introduced in the literature:
\begin{enumerate}
    \item[\textbf{C1}]  A Hawkes process is stable if $\rho(\abs(\mathbf{K})) <1$ \citep{bremaud_stability_1996}.  \customlabel{c1eigen}{\textbf{C1}}
    \item[\textbf{C2}] A Hawkes process is stable if $max_j \sum_{i = 1}^M K^+_{ij} < 1$ \citep{sulem_bayesian_2021}. \customlabel{c2colsum}{\textbf{C2}}
\end{enumerate}

Note that \ref{c1eigen} uses the spectral radius of the absolute value matrix, whereas \ref{c2colsum} utilises only the positive part of the matrix $\mathbf{K}$. Both conditions are sufficient, but not necessary. \citet{sulem_bayesian_2021} discuss both \ref{c1eigen} and \ref{c2colsum}, but do not compare them as neither is stronger than the other. As \ref{c1eigen} uses the absolute value matrix, negative entries (inhibition) are converted into excitation. Hence, while the intensity is decreased by certain events, stability is check as if those events increased the intensity. Of course, this procedure is sufficient, but the intensity $\lambda_m(\cdot | \abs(\mathbf{K}))$ may be rather different than the original $\lambda_m(\cdot | \mathbf{K})$, in particular when strong inhibition is present. Hence, we found there to be scope to develop a new criterion that is more closely tailored to processes that contain inhibition.

\subsection{Introducing a New Condition} \label{subsec_newcond}

We propose a new condition to assess stability for a given parameter $\mathbf{K}$:

\begin{theorem}[\textbf{C3}]\label{theorem_eigenplus} 
If $\rho(\mathbf{K}^+) < 1$, then the process with intensities $\lambda_m(\cdot \, | \, \mathbf{K}$), $m = 1 \dots M$, is stable.  \customlabel{c3new}{\textbf{C3}}
\end{theorem}

To prove it, we state an auxiliary lemma.

\begin{lemma} \label{lemma_thinning}
Suppose a process with intensity $\bar{\lambda}(\cdot)$ has finite average intensity. Then a process with intensity $\lambda(\cdot)$ such that $\bar{\lambda}(x) \geq \lambda(x)$ for all $x \geq 0$ also has finite average intensity. 
\end{lemma}

\begin{proof}
The stable process with intensity $\bar{\lambda}(\cdot)$ has a finite average intensity and $\bar{\lambda}(x) \geq \lambda(x)$ holds for all $x \geq 0$. The average intensity of $\lambda(\cdot)$ is therefore at most as large as the average intensity of $\bar{\lambda}(\cdot)$. Hence, the average intensity of $\lambda(x)$ is also finite. 
\end{proof}

Using this lemma we now prove Theorem~\ref{theorem_eigenplus}.

\begin{proof}
Since $K^+_{ij} \geq K_{ij}$ for each $i, j = 1 \dots M$, it follows that $\lambda_m (t \, | \, \mathbf{K}^+) \geq \lambda_m (t \, | \, \mathbf{K)}$ for all $m = 1 \dots M$ and all $t \in [0, T_{max}]$. By Lemma~\ref{lemma_thinning}, if a process with intensity $\lambda_m (\cdot \, | \, \mathbf{K}^+)$ has finite intensity then a process with intensity $\lambda_m (\cdot \, | \, \mathbf{K)}$ has finite intensity as well. A multivariate process has finite intensity if the intensity in each dimension $m = 1 \dots M$ has finite average. If $\rho(\mathbf{K}^+) < 1$ then the multivariate process with intensities $\lambda_m (\cdot \, | \, \mathbf{K}^+)$, $m = 1 \dots M$, has finite intensity, and therefore the multivariate process with intensities $\lambda_m (\cdot \, | \, \mathbf{K)}$, $m = 1 \dots M$, has a finite average intensity as well. Uniqueness is guaranteed by Theorem 6.55 from \citet{liniger_multivariate_2009}, and hence the process  $\lambda_m (\cdot \, | \, \mathbf{K)}$, $m = 1 \dots M$ is stable.
\end{proof}


\subsection{Comparison}

We now compare \ref{c3new} to either of the existing conditions (\ref{c1eigen}, \ref{c2colsum}) and show that if at least one of them holds, so does \ref{c3new}. Moreover, there are examples where only \ref{c3new} holds. This implies that \ref{c3new} can confirm stability for more parameters, which is useful when fitting a multivariate Hawkes process. 

\begin{theorem} \label{therorem_c1_c3}
 When \ref{c1eigen} holds, then \ref{c3new} holds as well.
\end{theorem}

We first state the following lemma.
\begin{lemma} \label{lemma_colsum_gelfand}
Let $X$ be a $N \times N$ matrix with entry $X_{ij}$ in row $i$ and colum $j$. Then  $  \rho(X) \leq \max_j \sum_{i = 1}^N X_{ij}$. 
\end{lemma}
This is a direct consequence of the Gelfand formula \citep{gelfand_normierte_1941}.

With that, we can now provide the proof for the stated Theorem~\ref{therorem_c1_c3}.

\begin{proof}
First we compare \ref{c3new} to \ref{c1eigen} when all entries $K_{ij}$ are non-negative (i.e. excitation only). It is trivial to see that when \ref{c1eigen} holds, then \ref{c3new} holds as well since $\abs(\mathbf{K}) = \mathbf{K} = \mathbf{K}^+$ and therefore $\rho(\mathbf{K}^+) < 1$. 

When we do not restrict the entries $K_{ij}$ to non-negative, we note that each entry of $\abs(\mathbf{K})$ is as least as large as the corresponding entry in $ \mathbf{K}^+$. Within the entry-wise positive matrices, the spectral radius is monotonous \citep[Lemma 12, p. 153][]{serre_matrices_2002}. Hence, if $\rho(\abs(\mathbf{K})) < $ then also $\rho(\mathbf{K}^+) < 1$.

Therefore, when \ref{c1eigen} holds, then \ref{c3new} holds as well.
\end{proof}

\begin{theorem} \label{therorem_c2_c3}
 When \ref{c2colsum} holds, then \ref{c3new} holds as well.
\end{theorem}

\begin{proof}
By Lemma~\ref{lemma_colsum_gelfand}, if $\max_j \sum_{i = 1}^N K_{ij} <1$ then also $\rho(\mathbf{K}^+) < 1$ and therefore \ref{c3new} holds if \ref{c2colsum} holds.
\end{proof}

Hence, we have shown that if at least one of (\ref{c1eigen}, \ref{c2colsum}) holds, so does \ref{c3new}. In addition, there are examples where neither of the existing conditions could confirm stability, but by using \ref{c3new} we can verify that the process is stable. Let us examine a two-dimensional Hawkes process with
\begin{align*}
    \mathbf{K} = \left( \begin{array}{cc}
        0.5 &  1\\
        -2 & 0.5
    \end{array} \right)
\end{align*}
as an illustrative example. Note that $\rho(\abs(\mathbf{K})) > 1$ and the maximum column sum of $\mathbf{K}^+$ is also larger than $1$, hence neither \ref{c1eigen} nor \ref{c2colsum} hold. Hence, by just using the two existing conditions it is not possible to assess whether a process using $\mathbf{K}$ would be stable. However we can make use of \ref{c3new}, as $\rho(\mathbf{K}^+) < 1$ and confirm that a process using $\mathbf{K}$ is stable. 

In summary, \ref{c3new} not only provides a unified approach to assess stability, it also permits us to determine stability for more parameters than by just using (\ref{c1eigen}, \ref{c2colsum}).


\section{Application} \label{sec_application_real}

This section studies two examples of a multivariate Hawkes process with excitation and inhibition to detect product cannibalisation using the model from Section~\ref{subsec_model} and the priors from Section~\ref{sec_priorchoice}. The first example models two products from product class A, whereas the second one examines how the orders four products in a different product class B interact. For data privacy reasons we cannot disclose the nature of these product classes. For each example we also fit two models without inhibition that serve as comparisons both on the training and test set. 

\subsection{Product Class A, $M=2$} \label{sec_app_M2_A}

For our first example we select two similar products to examine the product cannibalisation between them. Both are similar in their appearance and target audience. Their suggested retail prices are also differs approximately $20\%$.

We use one year (2016-06-14 to 2017-06-13) as training data (a total of $109$ observations, $55$ for Article $1$, $54$ for Article $2$). The following half year (2017-06-14 to 2017-12-13) is used as a test period ($91$ events, of which $60$ from Article $1$ and $31$ from Article $2$). Figure~\ref{fig_M2_observations} displays these events. As discussed in Section \ref{sec_data}, we are dealing with wholesale data, hence all sales are to the same wholesale customer (BusinessGroup).

\begin{figure}[tp]
    \centering
    \includegraphics[width=.99\linewidth]{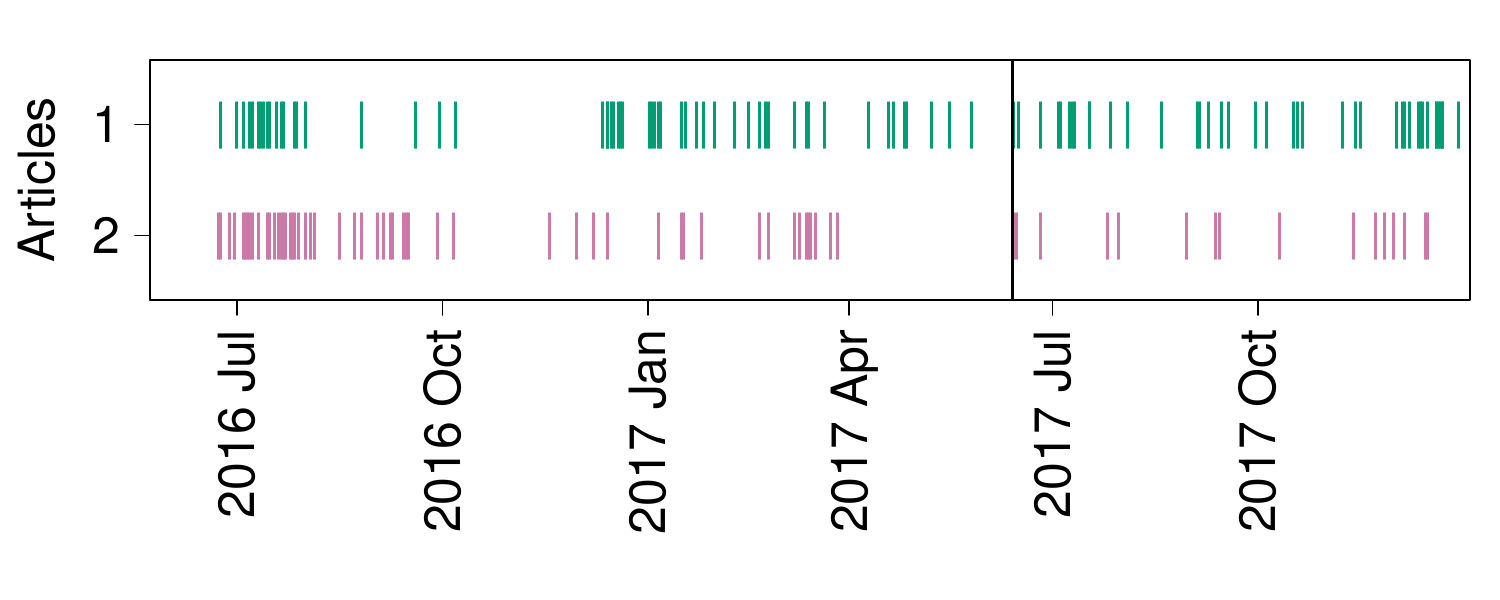}
    \caption{Orders placed for two products from product class A by BusinessGroup. Each vertical bar indicates that on the particular day an order was placed that included the respective article. The vertical black line indicates the split between training and test data.}
    \label{fig_M2_observations}
\end{figure}

As described in Section~\ref{subsec_est_backgr}, the plug-in estimate for $b(\cdot)$ of the background rate is based on all articles in the same product class. This ensures that $b(\cdot)$ only captures large, seasonal trends. Given this plug-in estimate we then use the prior set up from Section~\ref{sec_priorchoice} to obtain posterior distributions for $\Theta$ using Stan \citep{Stan_RStan_2019}. The estimation is carried out on the training set using using normal priors that permit inhibition for $\mathbf{K}^*$ (Model 1). For comparison, we also fit two additional models without inhibition that serve as benchmarks. Model 2 does not allow any inhibition ($0 < K_{ij} <1$ for all $i,j$), whereas Model 3 only uses the background rate, which is equal to $ K_{ij} = 0$ for all $i,j$. Figure~\ref{fig_M2_posterior} showcases the posterior distribution for $\mathbf{K}^*$ for Model 1 and Model 2. Model 3 is not included in the plot as the posterior distributions is simply a point mass at zero for each entry.

\begin{figure}[tp]
    \centering
    \includegraphics[width=.99\linewidth]{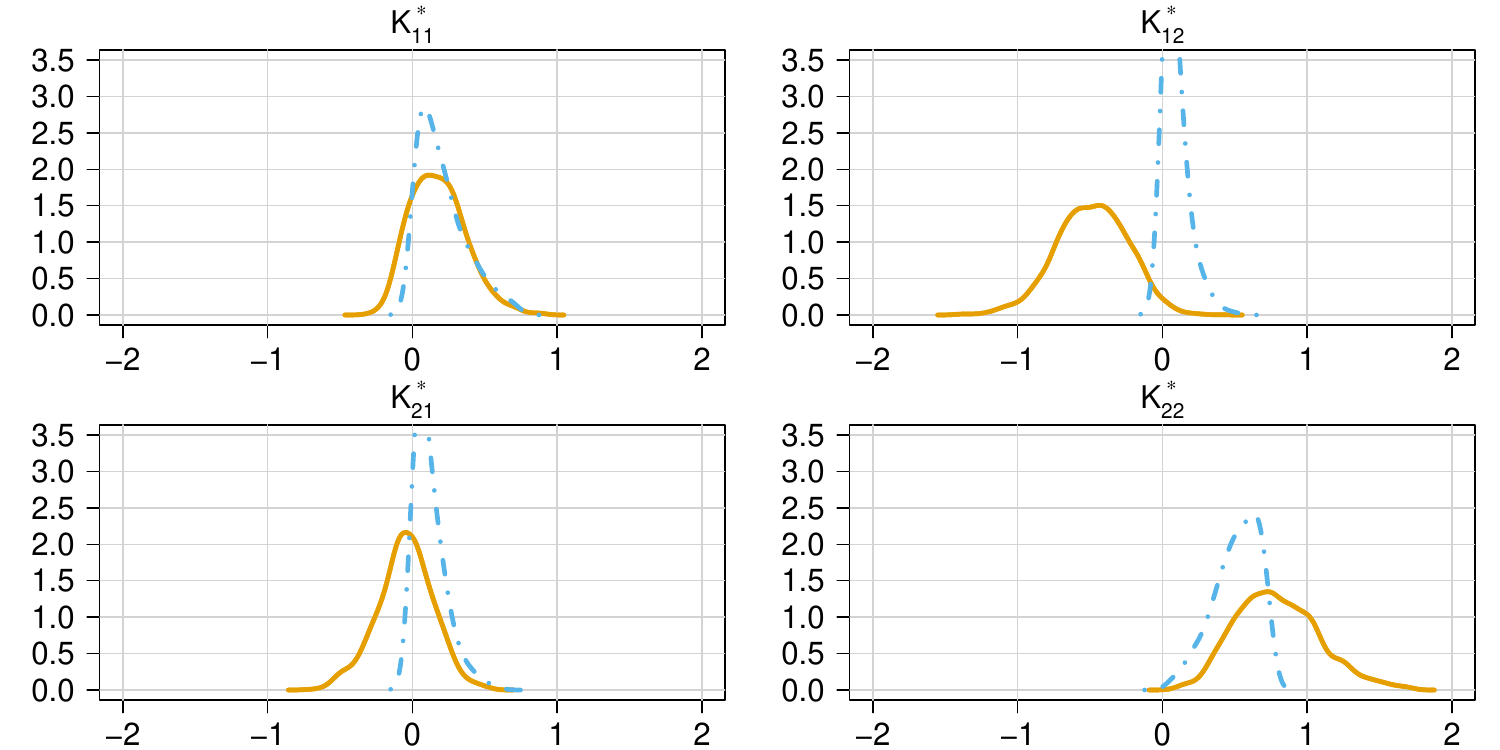}
    \caption{Posterior density estimates for the entries of $\mathbf{K}^*$ based on orders placed for two products Business on the training set. The orange solid line represents the normal priors (Model 1), the blue dot-dashed line is the excitation-only reference model (Model 2). }
    \label{fig_M2_posterior}
\end{figure}

For those model that allows inhibition (Model 1) the parameter $K^*_{12}$ is estimated to be negative. We estimate $K^*_{12} = -0.49$ with a $90\%$- credible interval of ($-0.90, -0.10$). We can interpret this as product cannibalisation in the sense that Article 1 cannibalises sales of Article 2. Interestingly, the other cross-influence parameter $K^*_{21}$ is estimated as (close to) zero. As the two articles are very similar from am appearance perspective, their main difference lies in the suggested retail price. This analysis suggest that the wholesale customer is buying the slightly cheaper article instead of the more expensive one, but not vice versa. 

In addition we can also examine how long the self and cross-influence last. The posterior means for in Model 1 are $\beta_{\text{diag}} = 0.14$ ($0.11,0.49$) and $\beta_{\text{off}}=0.33$ ($0.11,0.49$). This leads to to the influence kernels (for $K = 1$) as showcased in Figure~\ref{fig_M2_beta} for the self and cross-influence. Half of the self-influence takes place in the first five days after the event, whereas this number sits at two days for the cross-influence. This means that the cross-inhibitory effects are most pronounced immediately after an order was placed.

\begin{figure}[tp]
    \centering
    \includegraphics[width=.8\linewidth]{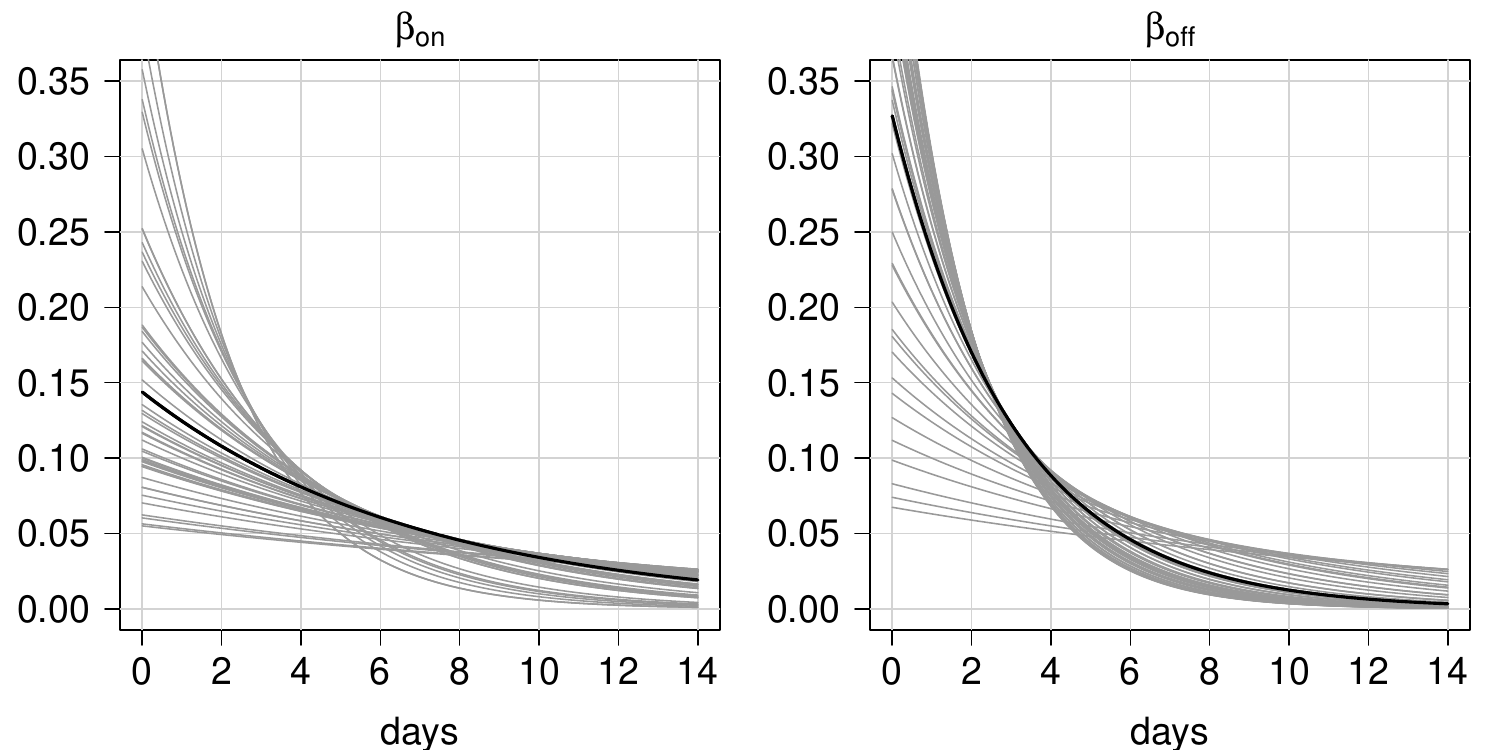}
    \caption{Plots of the influence kernels (self and cross-influence) based on the estimated of Model 1 on the training set for $M=2$ articles. Black line shows influence kernel at posterior mean, the grey lines show influence kernels from $50$ posterior samples. }
    \label{fig_M2_beta}
\end{figure}

Table~\ref{tab_comparisonM2} compares our suggested model incorporating inhibition (Model 1) to two benchmark models which do not allow inhibition. This comparison is done both both within the training set (training set log-likelihood and DIC), as well as the predictive likelihood on the test set. Across all comparisons the inhibition-encompassing Model 1 shows the best performance. Both in the training and test set the models without inhibition (Model 2 and Model 3) give worse outcomes. These results clearly show the need for inhibition and hence product-cannibalisation when modelling the sales process using point processes.

\begin{table}
\caption{\label{tab_comparisonM2} Comparisons of three models for $M=2$.} 
    \centering
    \begin{tabular}{clcccc}
    \hline
         &\textbf{Model} & \textbf{Restriction}   & \textbf{loglik train} & \textbf{DIC} &  \textbf{predictive loglik}  \\
         \hline
        1 & Normal Prior & $K_{ij} <1$ & \textbf{-273.39} & \textbf{557.44} & \textbf{ -169.12} \\
        2 & Excitation only & $0 < K_{ij} <1$  & -278.14 & 562.36 & -180.18  \\
        3 & Background only & $K_{ij} = 0$ &  -282.27 & 568.38 &   -172.34\\
         \hline
    \end{tabular}
  \end{table}

\subsection{Product Class B, $M= 4$} \label{sec_app_M4_B}

Our second example we look at the orders placed by BusinessGroup for four similar products from product class B. Table~\ref{tab_M4_products} gives some characteristics of these articles. We use one year as a training period (2016-06-14 to 2017-06-13) and the consecutive half year as a test period (2017-06-14 to 2017-12-13). The occurrences of events are displayed in Figure \ref{fig_M4_observations} and Table~\ref{tab_M4_products} gives the number of orders per article in the train and test set. 

\begin{table}
\caption{\label{tab_M4_products} Information on four articles used in the example in Section~\ref{sec_app_M4_B}. Columns 2-4 describe the articles. Last column contains the number of orders placed by BusinessGroup from 2017-06-14 to 2017-12-13 for this article in the training (test) period.} 
    \centering
    \begin{tabular}{cccccc}
    \hline
         \textbf{Article} & \textbf{Appearance}  & \textbf{Details} & \textbf{Label} & \textbf{Price} & \textbf{Orders train (test)} \\
         \hline
         1 & dark & white & none & low & 63 (11)  \\
         2 & light & colour & SomeLabel & high & 148 (94) \\
         3 & dark & minimal & SomeLabel & high & 30 (20) \\
         4 & light & minimal & SomeLabel & medium & 105 (7) \\
       
         \hline
    \end{tabular}
  \end{table}

\begin{figure}[tp]
    \centering
    \includegraphics[width=.99\linewidth]{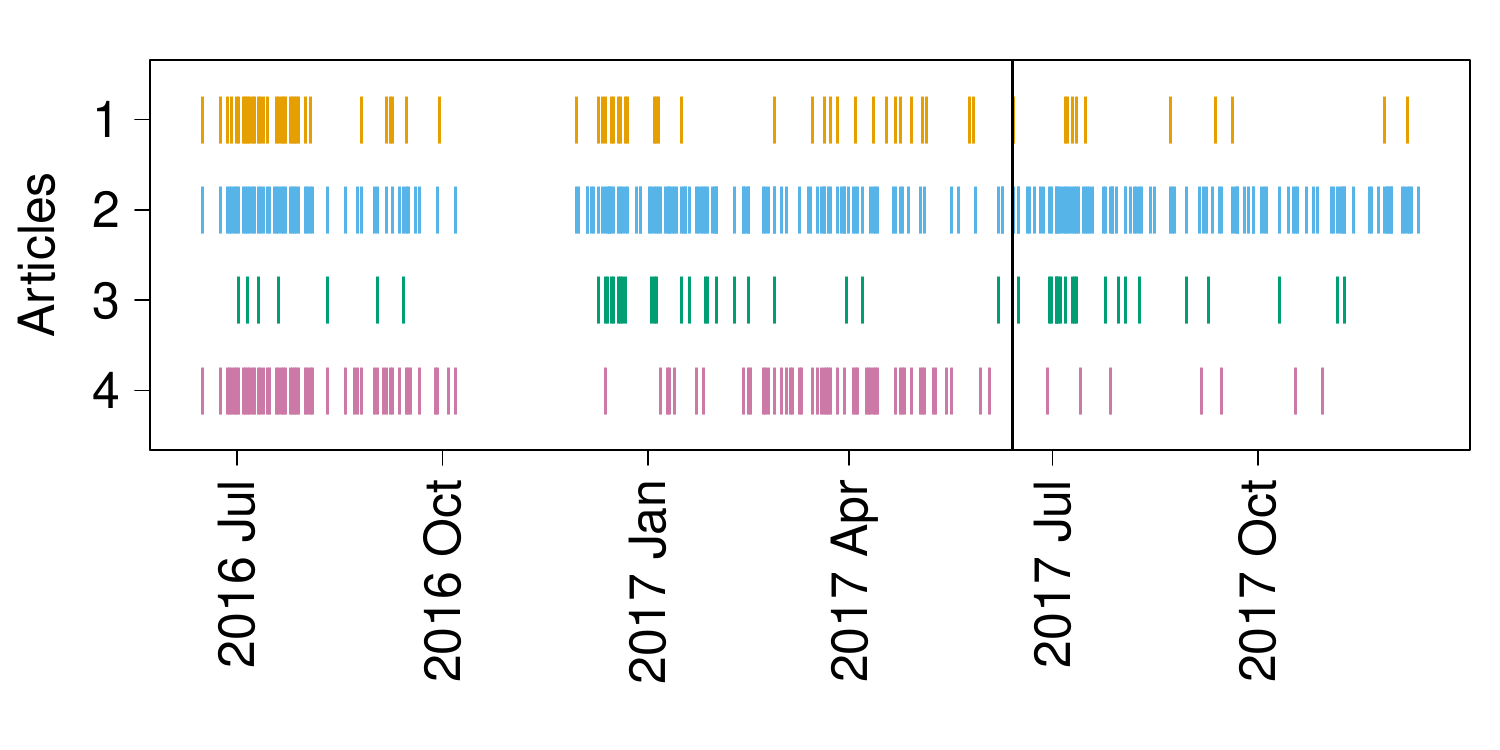}
    \caption{Orders placed for products in Product Class B by BusinessGroup. Each vertical bar indicates that on the particular day an order was placed that included the respective article. The vertical black line indicates the split between training and test data.}
    \label{fig_M4_observations}
\end{figure}

For the plug-in estimate for $b(\cdot)$ of the background rate we use all products in the same class that have sales in the relevant period. As in the above Section~\ref{sec_app_M2_A} we use the prior set up from Section~\ref{sec_priorchoice} to obtain posterior distribution samples for $\Theta$ using Stan \citep{Stan_RStan_2019}.

As outlined above, we fit one proposed models that incorporates inhibition and two benchmark models without inhibition. Model 1 uses normal priors which permit inhibition, whereas  Model 2 is excitation-only and Model 3 is only modelled by the scaled background rate. Figure~\ref{fig_M4_posterior} plots the obtained posterior distributions of the entries of $\mathbf{K}^*$ for Model 1 and Model 2. For Model 1, two parameter indicate product cannibalisation. We estimate $K^*_{32} = -0.60$ ($-1.15, 0.00)$ and $K^*_{43} = -0.40$ ($-0.82, 0.00$). We conclude that orders for Article 3 cannibalise orders of Article 2, as they are both at a higher price point and part of SomeLabel. In addition, orders for Article 4 cannibalise orders for Article 3. As above, we see that the cheaper article (with respect to the suggested retail price) cannibalises the more expensive one, but not vice versa. Also note that Article 1 is not affected by any inhibition, potentially due to the lower price point, the distinct design (white details), and lack of label affiliation, all of which seem to render it unsuitable to potentially substitute the other articles.

\begin{figure}[tp]
    \centering
    \includegraphics[width=.99\linewidth]{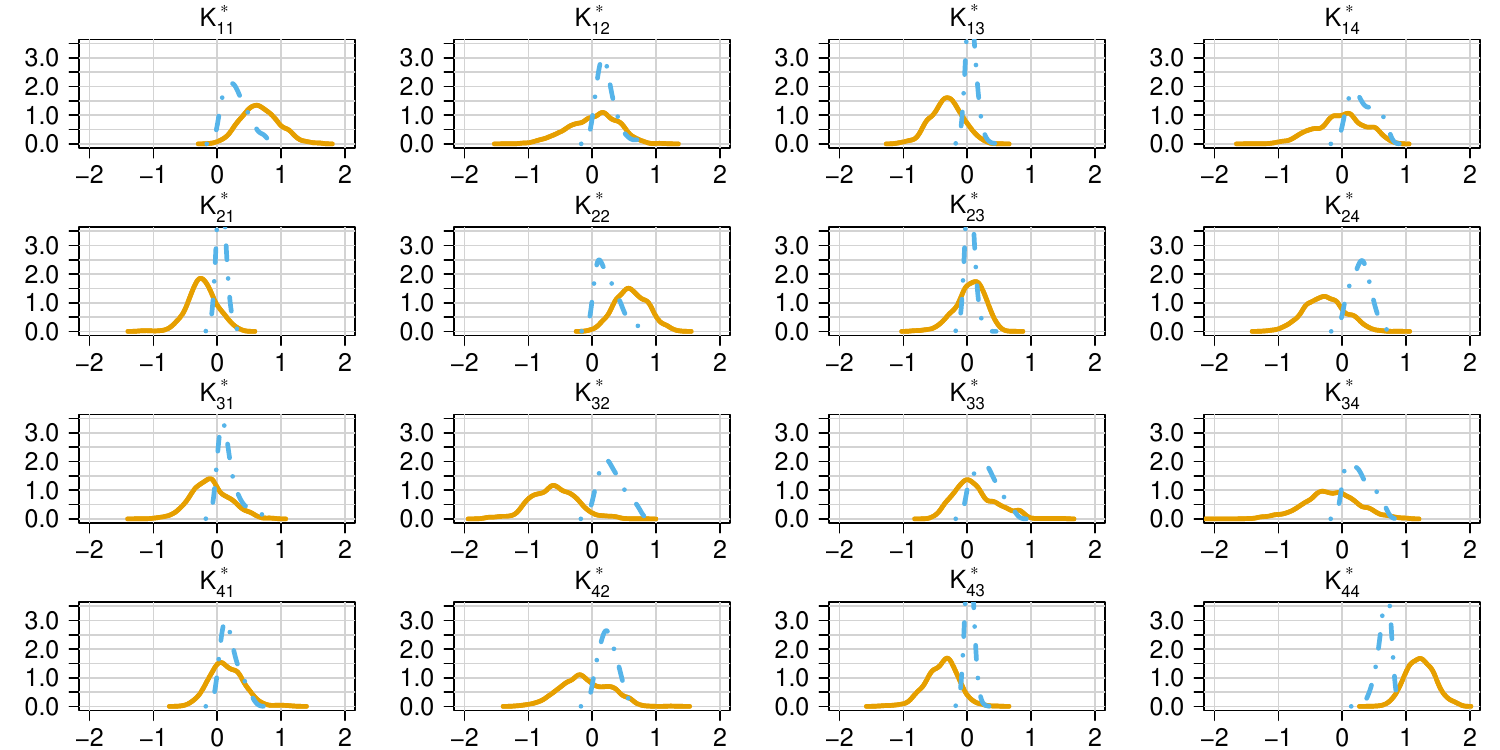}
    \caption{Posterior density estimates for the entries of $\mathbf{K}^*$ based on orders placed for products in product class B by BusinessGroup on the training set. The orange solid line represents the normal priors (Model 1), the blue dot-dashed line is the excitation-only reference model (Model 2). }
    \label{fig_M4_posterior}
\end{figure}

We also examine the shape of the influence kernels (both self and cross influence) in Figure~\ref{fig_M4_beta}. The self-influence kernel has a median of nine days and the cross-influence kernel's media lies at two days. This suggests again that any cross-influences, such as inhibition, are most influential in the days immediately after an order was placed.

\begin{figure}[tp]
    \centering
    \includegraphics[width=.8\linewidth]{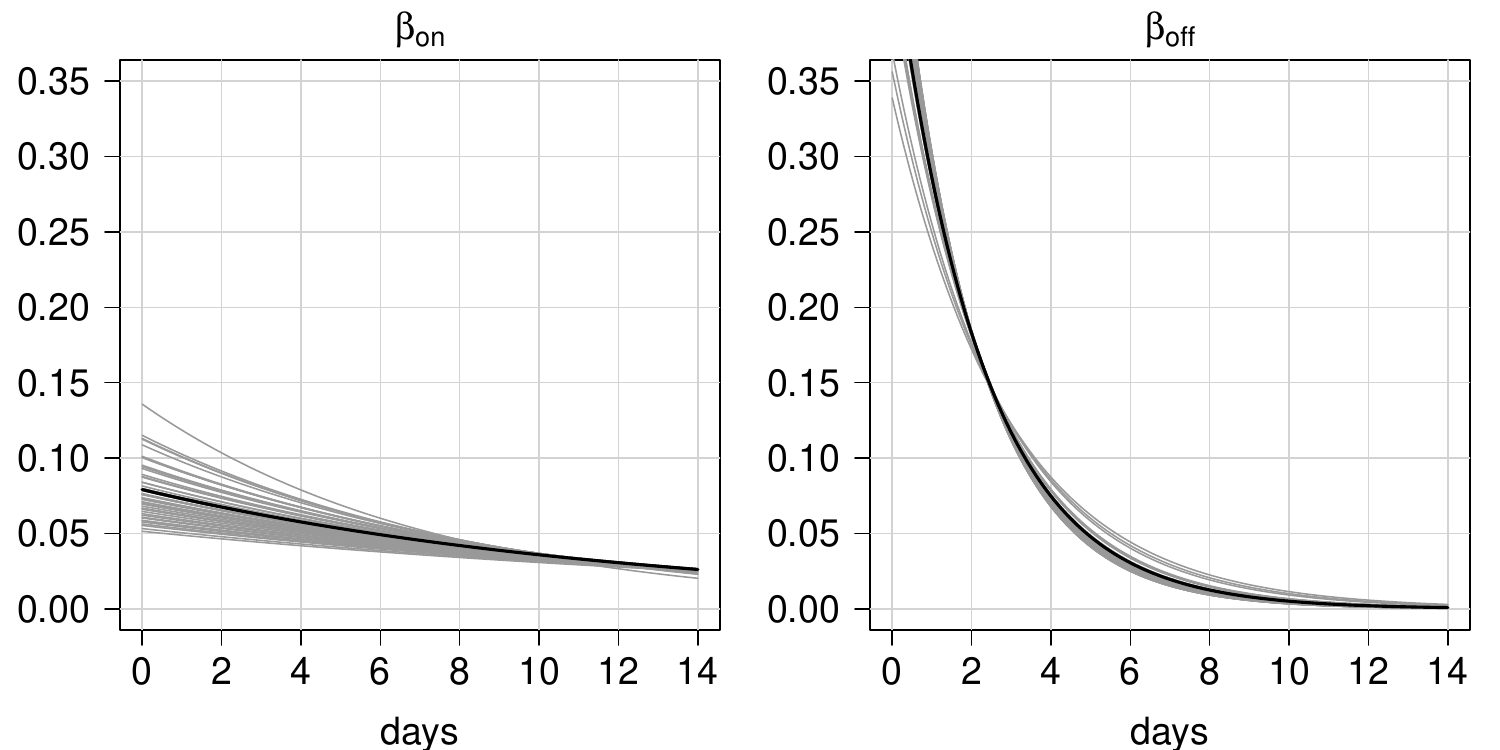}
    \caption{Plots of the influence kernels (self and cross-influence) based on the estimated of Model 1 on the training set for $M=4$ articles. Black line shows influence kernel at posterior mean, the grey lines show influence kernels from $50$ posterior samples. }
    \label{fig_M4_beta}
\end{figure}

Table \ref{tab_comparisonM4} compares the three fitted models both on the training set (using the log-likelihood as well as the DIC) and the test set (predictive log-likelihood). Model 1 with inhibition achieves the highest log-likelihood and lowest DIC the training set, as well as the highest predictive likelihood on the test data. As above, models without inhibition (Model 2 and Model 3) are distinctively sub-par across all metrics. This clearly warrants the consideration of inhibition and hence product cannibalisation when modelling the orders for similar articles. 

\begin{table}
\caption{\label{tab_comparisonM4} Comparisons of three models for $M=4$.} 
    \centering
    \begin{tabular}{clcccc}
    \hline
         &\textbf{Model} & \textbf{Restriction}   & \textbf{loglik train} & \textbf{DIC} &  \textbf{predictive loglik}  \\
         \hline
        1 & Normal Prior & $K_{ij} <1$  & \textbf{-678.73} & \textbf{1381.45} & \textbf{-297.34} \\
        2 & Excitation only & $0 < K_{ij} <1$  & -698.10 & 1411.39 & -313.32 \\
        3 & Background only & $K_{ij} = 0$  & -726.33 & 1460.49 &   -349.23\\
         \hline
    \end{tabular}
  \end{table}

\section{Discussion} \label{sec_discussion}

In this paper we focused on estimating product cannibalisation for wholesale data using the multivariate Hawkes process. On one hand we provided a statistical model to estimate product cannibalisation. On the other hand we made it easier for anyone to use multivariate Hawkes processes with inhibition through our considerations of a non-negative intensity and stability under inhibition.

We used our proposed model to estimate product cannibalisation for $M =2$ and $M=4$ articles and compared our suggested estimation incorporating inhibition to two reference models without inhibition. The superior performance of the models with inhibition across the board gives a strong mandate to consider product cannibalisation when modelling wholesale orders. 

Our work can be extended by considering different link functions and a variety of influence kernels, or differently structured inhibition altogether \citep[for example along the lines of ][]{apostolopoulou_mutually_2019}. Moreover, all stability conditions, including our newly proposed \ref{c3new}, are only sufficient to check stability. It would be of great interest to develop a criterion that was both necessary and sufficient. 

We would also like to extend this work to higher dimensions and to incorporate additional covariates, as our data is rich in articles and information about them. This would warrant both computational considerations (e.g. introduce an upper limit for the influence to ease the likelihood calculation), as well as structural ones (e.g. regularisation). A logical next step is the inclusion of price or order size, which parallels the idea of marked point processes in the earthquake literature \citep[see, for example, ][]{schoenberg_2003_multidimensional}, that lay beyond the scope of this paper. As the data is recorded in days, we would also like to explore whether a multivariate Hawkes process in discrete time \citep[as used by][]{browning_simple_2021} could bring additional benefits. 

Finally, this work represents a new way to estimating product cannibalisation, in particular for the wholesale perspective. We therefore aim to apply our method to different scenarios from a variety of industries to explore where business insights can be generated.


\begin{appendix}
\section{Roots of the Intensity Function} \label{appendix_roots}


This section shows how the roots of the intensity function can be found when using the exponential kernel. As described in Section~\ref{subsec_integrate_intensity}, these can be used to calculate the integral of the intensity exactly for a multivariate Hawkes process with inhibition. \citet{bonnet_maximum_2021} provide an exact integral for the one-dimensional case. In addition, the work by \citet{bonnet_inference_2022} gives similar results to ours (thought they examine the case where $\beta_{ij} = \beta_{j}$), but was first submitted to the ArXiv in May 2022. However, an earlier version of this paper (including the results below) was available on the ArXiv from January 2022 onward.

\subsection{Exponential Kernel with All $\beta_{ij}$ Equal}  \label{appendix_roots_beta}

We start with a special case where all $\beta_{ij}$ are equal and examine the roots both in one dimension and $M$ dimensions. 

\subsubsection{One Dimension}
Assume that we have data $Y = (t_1 \dots t_N)$ from one dimension and the intensity function is defined with the exponential kernel: 
\begin{equation}
    \lambda(t) = \mu + \sum_{i: t > t_i} K \, \beta  \, exp\left(-\beta(t - t_i)\right)
\end{equation}

The intensity can only drop below zero when an event happens. Therefore it is sufficient to check only intervals after an event at which the intensity is negative to see if the intensity becomes positive again before the next event happens.

Then we can find the root $t$ between observation $t_n$ and $t_{n+1}$ in the following way:
\begin{align}
    \mu + \sum_{i = 1}^n K \, \beta  \, exp\left(-\beta(t - t_i)\right)  = 0 \\
    \mu + \sum_{i = 1}^n K \, \beta  \, exp\left(-\beta t \right) \, exp\left(\beta t_i\right) = 0 \\
    exp\left(-\beta t \right) \left[\sum_{i = 1}^n K \,\beta exp\left(\beta t_i\right) \right]  = -\mu \\
    t = \frac{log\left( \frac{-\mu}{\sum_{i = 1}^n K \,\beta exp\left(\beta t_i\right)}\right)}{- \beta}
\end{align}
if $t_n < t < t_{n+1}$.

\subsubsection{$M$ Dimensions}

Now assume we have $M$ dimensional data $Y_1 = \left(t_{1\,1} \dots t_{1\,N_1} \right) \dots Y_M = \left(t_{M\,1} \dots t_{M\,N_M} \right)$ with intensity function
\begin{equation}
     \lambda_i(t) = \mu_i + \sum_{j = 1}^M\sum_{q: t > t_{q\,l}} K_{ji}\, g_{ji}(t-t_{j\,l})
\end{equation}

Again, we only intervals after an event at which the intensity would be negative to see if the intensity becomes positive again before the next event happens. However, we need to check in each interval in each dimension. 

To check in dimension $m$ after an observation $t_n$ from arbitrary dimension we can start with
\begin{align}
    \mu_m + \sum_{q = 1}^M \sum_{i: t_{qi} < t_n} K_q \, \beta  \, exp\left(-\beta(t - t_{qi})\right)  = 0
\end{align}
which leads to the following expression for the root
\begin{align}
    t = \frac{log\left( \frac{-\mu_m}{\sum_{q = 1}^M \sum_{i: t_{q\,i} < t_n} K_q \,\beta exp\left(\beta t_{q\,i}\right)}\right)}{- \beta}
\end{align}
if $t_n < t < t_{n+1}$.

\subsection{General Exponential Kernel} \label{appendix_roots_general}
Now we examine the general case when all betas are different. Here, finding the roots is a polynomial problem:
\begin{align}
    \mu_m +  \sum_{q = 1}^M \sum_{i: t_{qi} < t_n} K_q \, \beta_q  \, exp\left(-\beta_q(t - t_{qi})\right) & = 0 \\
     \mu_m +  \sum_{q = 1}^M \sum_{i: t_{qi} < t_n} K_q \, \beta_q  \, u^{-\beta_q}  \, exp\left(\beta_q t_{qi}\right) & = 0 \text{ where } x = exp(t) \\
     \mu_m +  \sum_{q = 1}^M x^{-\beta_q}   \underbrace{\sum_{i: t_{qi} < t_n} K_q \, \beta_q  \,  exp\left(\beta_q t_{qi}\right)}_{v_q} & = 0 \\
     \mu_m +  \sum_{q = 1}^M x^{-\beta_q} v_q & = 0
\end{align}
This is now a polynomial in $x$ that needs to be solved.


\section{Approximating $\Lambda$ using the Simpson's Rule} \label{appendix_integral}

This section describes the approximation procedure for $\Lambda= \sum_{m = 1}^M \int_0^{T_{max}} \lambda_m(x  ) \, dx$ as mentioned in Section~\ref{subsec_integrate_intensity}. We use the Cubic Simpson's Rule where 
\begin{equation}
    \int_a^b f(x) \, dx \approx \frac{(b - a)}{8} \left[f(a) + 3f\left(\frac{2a + b}{3}\right) + 3f\left(\frac{a + 2b}{3}\right) + f(b)\right]
\end{equation}
Assume we have $M$ dimensional data $Y_1 = \left(t_{1\,1} \dots t_{1\,N_1} \right) \dots Y_M = \left(t_{M\,1} \dots t_{M\,N_M} \right)$ and intensity function $\lambda_m(\cdot)$ for each dimension $m = 1 \dots M$. Algorithm~\ref{alg_Lambda_Simpson} outlines the approximation procedure used. 

\begin{algorithm}[htb!]
\caption{Approximating $\Lambda$ using the Simpson's Rule}\label{alg_Lambda_Simpson}
\begin{algorithmic}[1]
\State Set $Y_{total}  = (0, t_{1\,1} \dots t_{1\,N_1} \dots t_{M\,1} \dots t_{M\,N_M}, T_{max})$ 
\State Order all entries of $Y_{total}$ and call the result $X = (x_1 < \dots < x_P)$ where $P = 2 + \sum_m^M N_m$
\State Define $X^t = \{x_i : x_i \leq t \}$ 
\State Set $res = 0$
\For{$i$ in $1:(P-1)$}
    \State Set $a = x_i$
    \State Set $b = x_{i+1}$
    \For{$m$ in $1:M$}
        \State Set res = res + $\frac{(b - a)}{8} \left[\lambda_m(a | X^a) + 3 \lambda_m\left(\frac{2a + b}{3} | X^a\right) + 3 \lambda_m\left(\frac{a + 2b}{3}| X^a\right) + \lambda_m (b| X^a)\right]$
    \EndFor
\EndFor
\State Return $res$ as the approximation of $\Lambda$
\end{algorithmic}
\end{algorithm}


\section{Total Number of Offsprings} \label{appendix_offsprings}

This section contains the proof of the total number of offsprings from Section~\ref{sec_offsprings}.

We write $K_{ij}^*$ is the total number of offsprings in dimension $j$ that are produced by an event cascade started by an immigrant event in $i$. This can be interpreted as the \textit{marginal} influence from $i$ onto $j$ as both direct and indirect influences are taken into account. We write $\mathbf{K}^*$ where $K_{ij}^*$ is the entry in row $i$ and column $j$. 

Here, $N = (N_1 \dots N_M)^T = $ is the total number of events in each dimension, which can be calculated by $ N = (I - \mathbf{K}^T)^{-1} \mathbf{\mu} $, as described by \citet{hawkes_spectra_1971, jovanovic_cumulants_2015}.

\begin{lemma}
\label{lemma_Ni}
The number of events in each dimension is 
\begin{equation}
    N = (N_1 \dots N_M)^T = \left(\mathbf{K}^{*} + I\right) ^T \mathbf{\mu} \nonumber
\end{equation}
where $I$ is the identity matrix of appropriate dimension $M$. 
\end{lemma}

\begin{proof}
The Hawkes process can be written as the superposition of Poisson processes as outlined by the branching structure interpretation (see Section~\ref{subsec_branching}). Hence, each $N_i$ consist of two parts:
\begin{enumerate}
    \item The number of immigrants events from the background process with rate $\mu_i$
    \item The number of offsprings in $i$ from an immigrant event in each dimension $j$
\end{enumerate}
Therefore, we can write
\begin{equation}
    N_i = \mu_i T_{max} + \sum_{j = 1}^M \mu_j T_{max}K_{ji}^* 
\end{equation}
In matrix notation this is
\begin{equation}
    N = (N_1 \dots N_M)^T = \left(\mathbf{K}^{*} + I\right) ^T \mathbf{\mu} \label{eq_N_kstar}
\end{equation}
where $I$ is the identity matrix of appropriate dimension $M$. 
\end{proof}

\begin{theorem}
The total number of offsprings is 
\begin{equation}
    \mathbf{K}^{*} = \left(I - \mathbf{K} \right) ^{-1} - I   \nonumber
\end{equation}
\end{theorem}

\begin{proof}
We use the definition of $N$ from the literature
\begin{equation}
   N = (I - \mathbf{K}^T)^{-1} \mathbf{\mu} 
\end{equation}
and relate it to Lemma~\ref{lemma_Ni}:
\begin{equation}
    (I - \mathbf{K^T})^{-1} \mathbf{\mu}  = \left(\mathbf{K}^{*} + I\right) ^T \mathbf{\mu}
\end{equation}
Rearranging, this allows us to write $\mathbf{K}^{*}$ as
\begin{align}
    \mathbf{K}^{*} & =  \left(I - \mathbf{K} \right) ^{-1} - I  \label{eq_Kstar}
\end{align}
\end{proof}

\end{appendix}

\bibliographystyle{apalike}

\bibliography{biblio}
\end{document}